\documentclass[11pt]{article}

\usepackage[T1]{fontenc}
\usepackage{amssymb,amsmath,amsfonts}
\usepackage{graphicx,xcolor,enumitem}
\usepackage{epsfig}
\usepackage{amsthm}
\usepackage{bm}
\usepackage{subcaption}
\usepackage[round]{natbib}
\usepackage{csquotes}
\renewcommand{\mkbegdispquote}[2]{\itshape}
\usepackage[a4paper, twoside, hmarginratio=1:1, vmarginratio=1:1, left=1in,top=1in]{geometry}

\RequirePackage[breaklinks=true, hidelinks]{hyperref}
\usepackage{breakcites}

\usepackage{algorithm, algorithmic}

\usepackage{accents}

\newcommand{\cA}{\mathcal{A}}
\newcommand{\cB}{\mathcal{B}}
\newcommand{\E}{\mathbb{E}}
\newcommand{\R}{\mathbb{R}}
\newcommand{\p}{\mathbb{P}}
\newcommand{\cL}{{\mathcal L}}

\newcommand{\cF}{{\mathcal F}}
\newcommand{\cP}{{\mathcal P}}
\newcommand{\cT}{{\mathcal T}}
\newcommand{\id}{\text{id}}

\newcommand{\cV}{{\mathcal V}}
\newcommand{\cW}{{\mathcal W}}

\newcommand{\cX}{{\mathcal X}}
\newcommand{\cY}{{\mathcal Y}}

\newcommand{\cG}{{\mathcal G}}

\newcommand{\cR}{{\mathcal R}}
\newcommand{\cM}{{\mathcal M}}

\newtheorem{theorem}{Theorem}

\newtheorem{assumption}[theorem]{Assumption}

\newtheorem{corollary}[theorem]{Corollary}

\newtheorem{definition}[theorem]{Definition}

\newtheorem{lemma}[theorem]{Lemma}

\theoremstyle{definition}

\numberwithin{equation}{section}
\numberwithin{theorem}{section}

\begin{document}

\title{Distributionally robust risk evaluation with a causality constraint and structural information}

\author{Bingyan Han \thanks{Thrust of Financial Technology, The Hong Kong University of Science and Technology (Guangzhou), Email: bingyanhan@hkust-gz.edu.cn. Bingyan Han is partially supported by The Hong Kong University of Science and Technology (Guangzhou) Start-up Fund G0101000197, the Guangzhou-HKUST(GZ) Joint Funding Program (No. 2024A03J0630), and the National Natural Science Foundation of China (Grant No. 12401621).}
}

\date{May 26, 2025}
\maketitle

\begin{abstract}
	This work studies the distributionally robust evaluation of expected values over temporal data. A set of alternative measures is characterized by the causal optimal transport. We prove the strong duality and recast the causality constraint as minimization over an infinite-dimensional test function space. We approximate test functions by neural networks and prove the sample complexity with Rademacher complexity. An example is given to validate the feasibility of technical assumptions. Moreover, when structural information is available to further restrict the ambiguity set, we prove the dual formulation and provide efficient optimization methods. Our framework outperforms the classic counterparts in the distributionally robust portfolio selection problem. The connection with the naive strategy is also investigated numerically.
	\\[2ex] 
	\\[2ex] 
\end{abstract}

\section{Introduction}
The choice of the underlying probability measure is crucial in measuring risk or other objectives, but it presents a challenge for decision-makers (DMs) to find an expressive yet tractable reference measure. One simple option is to use the empirical measure on a sample dataset, but this choice is prone to model misspecification due to small sample sizes or blurred observations. Additionally, time-series data can pose further difficulties for modeling, creating another source of misspecification. This paper aims to provide a robust framework for risk evaluation with temporal data.

Since \cite{knight1921risk} clarified the subtle difference between risk and model uncertainty (misspecification), several methodologies have been developed to incorporate robustness. Typically, the objective is evaluated over a set of alternative plausible measures instead of a single measure. \cite{hansen2001robust} pioneered the robust control paradigm when alternative measures are equivalent to the reference measure, and relative entropy quantifies the difference between two measures. Along this direction, \cite{uppal2003model,maenhout2004robust,ben2010soft,han2021robust} investigated robust optimization in portfolio selection and asset pricing problems; see \cite{hansen2014nobel} for a review on model uncertainty. One limitation of relative entropy, however, is that it cannot compare non-equivalent measures when ambiguous volatility appears. Nonlinear expectations \citep{peng2010nonlinear} and McKean-Vlasov dynamic programming \citep{ismail2019robust} are two approaches to tackle the volatility uncertainty. The third paradigm utilizes Wasserstein distance and optimal transport (OT) \citep{villani2009optimal}. Wasserstein distance is the minimum cost incurred by an optimal coupling, known as an OT plan, that can reshape the reference measure to the alternative measure, possibly with different supports.

 % Review and advantages of OT methods.
OT enjoys intuitive geometric interpretation and has found interdisciplinary applications in mathematics, machine learning, statistics, and economics; see \cite{villani2009optimal,santambrogio2015optimal,galichon2016optimal,peyre2019computational} for a book-level introduction. One can embed flexible structures into the cost function of OT. Besides the modeling flexibility, OT also demonstrates analytical tractability and can be solved in explicit forms or with efficient algorithms. OT has been applied in model uncertainty, with early work by \cite{pflug2007ambiguity} defining the ambiguity set with the Wasserstein distance of order one. \cite{gao2016dist} considered distributionally robust stochastic optimization (DRSO) with ambiguity sets characterized by the Wasserstein distance of order $p$.  \cite{mohajerin2018data} proved that some DRSO problems have a finite convex program formulation. \cite{blanchet2019quantifying} allowed general lower semicontinuous cost functions and proved strong duality results with general Polish spaces as domains. \cite{guo2019mot,zhou2021distri} and references therein investigated martingale optimal transport with problems from option pricing and hedging. Applications in risk management and portfolio selection are considered by \cite{bartl2020computational,eckstein2020robust,blanchet2021dist}.

 % Limitations and motivate COT here.
However, to properly account for the chronological structure in temporal data, Wasserstein distance needs to be suitably formulated. Several researchers have proposed adapted versions of Wasserstein distance for this purpose. \cite{lassalle2013causal} introduced a causality constraint on the transport plans between discrete stochastic processes, termed as causal optimal transport (COT). \cite{pflug2012distance} defined a notion of {\it nested distance}, which can be roughly regarded as bicausal. A recent work by \cite{backhoff2020equal} showed that these definitions and many others define the same topology in finite discrete time. Roughly speaking, given the past of a process $X$, the past of another process $Y$ should be independent of the future of $X$ under the transport plan measure. COT is an emerging area that has found applications in mathematical finance \citep{backhoff2020adapted}, video modeling \citep{xu2020cot,xu2021quantized}, mean-field games \citep{acciaio2021cournot}, and stochastic optimization \citep{pflug2012distance,acciaio2020causal,bartl2022sensitivity}. In this paper, we study risk evaluation with temporal data and ambiguity sets characterized by COT. 
 
%Contributions and Findings.
Our contributions are twofold, encompassing both theoretical and practical aspects. First, we prove the strong duality theorem with lower semicontinuous costs. Compared with the OT paradigm in \cite{gao2016dist,blanchet2019quantifying}, the causality constraint adds an extra infinite-dimensional optimization problem. Motivated by the scalability and the universal approximation property of neural networks \citep{hornik1991,cybenko1989}, we use them to approximate the infinite-dimensional test function space. We further prove the sample complexity results with the classic Rademacher complexity theory, which can fully utilize the structure of the function class, especially the Lipschitz property of the activation function and the linear layers within the neural networks \cite[Theorem 12]{bartlett2002rademacher}. An explicit example of regularized one-layer neural networks is given to validate the feasibility of technical assumptions. 

To illustrate the fundamental difference between OT and COT, we present analytical examples before delving into general empirical experiments. In the extreme case, OT can entirely disregard the non-anticipativeness constraint, leading to distinct solutions compared with COT. On the other hand, when the cost and objective are separable, both OT and COT find the same worst-case scenarios. These scenarios are determined by the supremum of certain functions for each given sample, where the worst-case paths depend on the reference samples {\it pathwise and pointwise}. Notably, the temporal structure in the worst-case scenarios is identical to that of the reference samples and absorbs any potential noise in the empirical distributions. These observations suggest that OT and COT may include unrealistic alternative scenarios and be overly conservative.

To resolve this drawback, we propose a new framework, called Structural Causal Optimal Transport (SCOT), which allows the structural beliefs of the DM to be incorporated into alternative measures. The structural information helps to find worst-case scenarios that differ from the reference data and avoid conservative paths. This approach rules out some unrealistic scenarios in the Wasserstein ball. The duality theorem links the SCOT primal problem with calculations of the Wasserstein distance under modified costs. We also provide a finite sample guarantee for both SCOT and COT. 

As an application, we revisit the comparison between mean-variance portfolio selection and the naive strategy \citep{demiguel2009}. Due to deviations between the reference measure $\hat{\mu}$ and the true measure $\mu$, we enhance distributional robustness in mean-variance portfolios using OT and SCOT methods. Across nine out-of-sample cases studied, SCOT strategies consistently achieve higher Sharpe ratios than OT strategies in eight cases. Furthermore, SCOT strategies outperform both non-robust and naive approaches, especially under smaller radii. Extending conclusions from \cite{pflug2012naive} to a multi-period framework, our numerical results indicate that both OT and SCOT strategies converge towards the naive strategy under high model uncertainty, particularly when the investor is less risk averse.

The remainder of this paper is structured as follows. Section \ref{sec:form} introduces the background and formulates the primal problem of interest. Section \ref{sec:dual} proves the strong duality of COT formulation and motivates the SCOT framework. Section \ref{sec:complex} proves the sample complexity results and describes the optimization algorithms used. Section \ref{sec:num} presents the numerical analysis and compares different methodologies. Section \ref{sec:conclude} concludes the paper with further directions. Technical proofs are provided in the Appendix. The code used in this work is publicly available at \url{https://github.com/hanbingyan/SCOT_examples}.

\section{Problem formulation}\label{sec:form}

Let integer $T$ be the finite number of periods. Denote $[T] := \{1, 2, ..., T-1, T\}$. For each $t \in [T]$, suppose $\cX_t$ is a closed (but not necessarily bounded) subset of $\R^d$, under the Euclidean topology $\cT^{\cX_t}$. $\cX_t$ is interpreted as the range of the process at time $t$.
$\cX_{1:T} := \cX_1 \times \ldots \times \cX_T$ is a closed subset of $\R^{d \times T}$ and equipped with the product topology $\cT^\cX := \prod^T_{t=1} \cT^{\cX_t}$. We also use a notation $\cX := \cX_{1:T}$ for convenience. Let $C_b(\cX_{1:t})$ be the set of all continuous and bounded functions on $\cX_{1:t}$ and $\cP(\cX)$ be the set of all Borel probability measures on $\cX$. 

A DM wants to evaluate $\int_\cX f(x) \mu (dx)$ for a true but unknown probability measure $\mu \in \cP(\cX)$. The function $f$ is given with technical assumptions specified later. Following the convention in \cite{blanchet2019quantifying,gao2016dist}, we refer to the computation of $\int_\cX f(x) \mu (dx)$ as risk evaluation. Typically, the DM can only observe data from $\mu$ and constructs some reference measure $\hat{\mu}$. A typical example is the empirical measure. Then the DM calculates
 \begin{equation}\label{eq:risk}
 	\int_\cX f(x) \hat{\mu} (dx)
 \end{equation}
to approximate $\int_\cX f(x) \mu (dx)$.
 
An essential difficulty in risk evaluation is that the reference measure $\hat{\mu}$ can be misspecified. Commonly, it is difficult or even impossible to obtain a sufficiently large dataset. Besides, observations may exhibit a high level of noise, especially in financial data. To incorporate robustness into risk evaluation, an approach is to consider some close but different measures as alternatives to $\hat{\mu}$. A methodology is distributionally robust risk evaluation with an optimal transport framework. We review the paradigm as follows.

Denote $\cY := \cY_1 \times ... \times \cY_T$ as another closed subset of $\R^{d \times T}$, equipped with the Euclidean topology $\cT^\cY := \prod^T_{t=1} \cT^{\cY_t}$. $\cP(\cY)$ is the set of all Borel probability measures on $\cY$. Consider $\nu \in \cP(\cY)$ as an alternative measure. A coupling $\pi \in \cP(\cX \times \cY)$ is a Borel probability measure that admits $\hat{\mu}$ and $\nu$ as its marginals on $\cX$ and $\cY$, respectively. Denote $\Pi(\hat{\mu}, \nu)$ as the set of all the couplings. $\pi(x, y)$ is usually known as a transport plan between $\hat{\mu}$ and $\nu$. Heuristically, one can interpret $\pi(x, y)$ as the probability mass moved from $x$ to $y$. 

Suppose transporting one unit of mass from $x$ to $y$ incurs a cost of $c(x,y)$. The classic OT problem is formulated as 
\begin{equation}\label{eq:OT-dist}
	\cW(\hat{\mu}, \nu) := \inf_{ \pi \in \Pi(\hat{\mu}, \nu)} \int_{\cX \times \cY} c(x, y) \pi(dx, dy).
\end{equation}
Intuitively, $\cW(\hat{\mu}, \nu)$ quantifies how difficult it is to reshape the measure $\hat{\mu}$ to $\nu$. \cite{blanchet2019quantifying} consider the risk evaluation over alternative measures $\nu$ in the Wasserstein ball with a radius of $\varepsilon > 0$:
\begin{equation}\label{eq:OT}
	\sup_\nu \int_\cY f(y) \nu(dy), \text{ subject to }  \cW (\hat{\mu}, \nu) \leq \varepsilon.
\end{equation}
We refer to the problem \eqref{eq:OT} as the OT risk evaluation or simply the OT method (formulation). 

In this work, we consider the temporal data as $x = (x_1, ..., x_t, ..., x_T)$, which is common in statistics and finance. Then a natural requirement of the transport plan $\pi(x, y)$ is the non-anticipative condition. Indeed, if the past of $x$ is given, then the past of $y$ should be independent of the future of $x$ under the measure $\pi$. Mathematically, a transport plan $\pi$ should satisfy
\begin{equation}\label{eq:causal}
	\pi (dy_t | dx_{1:T}) = \pi (dy_t | dx_{1:t}), \quad  t = 1, ..., T-1, \quad \pi\text{-a.s.}
\end{equation}
The property \eqref{eq:causal} is known as the causality condition, and the transport plan satisfying \eqref{eq:causal} is called {\it causal} by \cite{lassalle2013causal}. Denote $\Pi_c(\hat{\mu}, \nu)$ as the set of all causal transport plans between $\hat{\mu}$ and $\nu$. 

In contrast to the OT formulation \eqref{eq:OT-dist}, the COT problem with temporal data restricts the minimization of the cost only over $\Pi_c(\hat{\mu}, \nu)$:
\begin{equation}\label{eq:causal-dist}
	\cW_c (\hat{\mu}, \nu) := \inf_{ \pi \in \Pi_c(\hat{\mu}, \nu)} \int_{\cX \times \cY} c(x, y) \pi (dx, dy).
\end{equation}
Denote the causal Wasserstein ball as
\begin{equation}\label{eq:causal_ball}
	\cB_\varepsilon := \{ \nu: \cW_c (\hat{\mu}, \nu) \leq \varepsilon \}.
\end{equation}
The distributionally robust risk evaluation under COT is 
\begin{equation}\label{eq:primal}
	J(\varepsilon; \hat{\mu}) := \sup_{\nu \in \cB_\varepsilon} \int_\cY f(y) \nu(dy).
\end{equation}
Since the causality condition may rule out some transport plans, one can expect that $\cW_c (\hat{\mu}, \nu) \geq \cW (\hat{\mu}, \nu)$. The causal Wasserstein ball may contain fewer alternative measures and discard them due to the causality. We call \eqref{eq:primal} the primal problem of the COT risk evaluation. The primal value may be less than the counterpart of the OT method.

\section{Duality}\label{sec:dual}

A key observation of the causality condition is that it can be formulated as a linear constraint; see \citet[Proposition 2.4]{backhoff2017causal}. In this section, we prove the strong duality theorem for \eqref{eq:primal}.

\subsection{Strong duality for COT risk evaluation}

\begin{assumption}\label{assum:lsc}
	The objective $f: \cX \rightarrow \R$ is upper semicontinuous (u.s.c.). The cost $c(x,y): \cX \times \cY \rightarrow [0, \infty)$ is lower semicontinuous (l.s.c.) and non-negative.
\end{assumption}
Compared with \citet[Assumption 1]{blanchet2019quantifying}, we drop the requirement that $c(x, x) = 0$ but do not allow $c(x, y)$ to be infinite. It becomes easier to design an analytical example later in Section \ref{sec:motivate} to present differences in several formulations. 

Under Assumption \ref{assum:lsc}, \citet[Theorem 4.3]{acciaio2021cournot} extended \citet[Theorem 2.6]{backhoff2017causal} and proved that the infimum of COT in \eqref{eq:causal-dist} is attained. One can rewrite the primal problem \eqref{eq:primal} as
\begin{equation*}
	\sup_{\pi} \int_{\cX \times \cY} f(y) \pi(dx, dy), \text{ subject to } \pi \in \cup_{\nu} \Pi_c(\hat{\mu}, \nu), \text{ and } \int_{\cX \times \cY} c(x,y) \pi(dx, dy) \leq \varepsilon.
\end{equation*}

As preparation for proving the strong duality, we introduce a test function space as 
\begin{align*}
	\Gamma := \Big\{ & \gamma: \cX \times \cY \rightarrow \R \, \Big| \, \gamma(x, y) =  \sum^L_{l = 1} \sum^{T-1}_{t=1} h_{l, t}(y_{1:t}) \Big[ g_{l, t}(x_{1:T}) \\
	& - \int_{\cX_{t+1:T}} g_{l, t}( x_{1:t}, x_{t+1:T}) \hat{\mu} (d x_{t+1:T} | x_{1:t}) \Big], \text{ for } h_{l, t} \in C_b(\cY_{1:t}), g_{l, t} \in C_b(\cX_{1:T}), \\
	& \text{ and some positive integer } L \in \mathbb{N} \Big\}. 
\end{align*}

Consider the dual problem
	\begin{equation}\label{eq:minimax}
		D(\varepsilon; \hat{\mu} ) := \inf_{ \lambda \geq 0, \; \gamma \in \Gamma} \quad  \lambda \varepsilon + \int_{\cX} F(x; \lambda, \gamma) \hat{\mu}(dx),
	\end{equation}
	where
	\begin{align*}
		F(x; \lambda, \gamma) := \sup_{y \in \cY}  \big\{& f(y) - \lambda c(x, y) + \gamma(x, y) \big\}.
	\end{align*}

As the first main result, Theorem \ref{thm:cap-dual} proves the strong duality for the primal problem \eqref{eq:primal}, with unbounded domains and a bounded objective $f$. The idea relies on \citet[Theorem 2.2 and Proposition 2.3]{bartl2019robust}, which gives a general representation of monotone convex functionals.

\begin{theorem}\label{thm:cap-dual}
	Suppose $\cX$ and $\cY$ are closed and Assumption \ref{assum:lsc} holds. Furthermore, 
	\begin{enumerate}
		\item $f$ is bounded;
		\item $\int_\cX c(x, x) \hat{\mu} (dx) \leq \varepsilon$;
		\item for every $k \geq 0$, there exists $r > 0$, such that $c(x, y) \geq k$ whenever $|x - y| \geq r$. 
	\end{enumerate}
	Then the strong duality holds as $J(\varepsilon; \hat{\mu}) = D(\varepsilon; \hat{\mu})$. Moreover, there exists a primal optimizer $\nu^* \in \cB_\varepsilon$ for $J(\varepsilon; \hat{\mu})$.
\end{theorem}

	With Theorem \ref{thm:cap-dual}, Corollary \ref{cor:dual} demonstrates that the strong duality remains valid when $f$ is bounded from above, but domains are compact.
\begin{corollary}\label{cor:dual}
	Suppose
	\begin{enumerate}
		\item Assumption \ref{assum:lsc} holds;
		\item $\cX$ and $\cY$ are compact;
		\item $\int_\cX c(x, x) \hat{\mu} (dx) \leq \varepsilon$.
	\end{enumerate}
	Then the strong duality holds as $J(\varepsilon; \hat{\mu}) = D(\varepsilon; \hat{\mu})$. Moreover, there exists a primal optimizer $\nu^* \in \cB_\varepsilon$ for $J(\varepsilon; \hat{\mu})$.
\end{corollary}

As a side note, the minimization over $g_{l,t}$ can be equivalently formulated as over martingales by \citet[Proposition 2.4]{backhoff2017causal}. Let $\cF^\cX$ be the filtration generated by the coordinate process on $\cX$. The corresponding test function space is 
\begin{align*}
	\Gamma' := \Big\{ & \gamma' : \cX \times \cY \rightarrow \R \, \Big| \, \gamma'(x, y) =  \sum^L_{l = 1} \sum^{T-1}_{t=1} h_{l, t}(y_{1:t}) \Big[ M_{l, t+1}(x_{1:t+1}) - M_{l, t}(x_{1:t})  \Big], \text{ with } \\
	&  h_{l, t} \in C_b(\cY_{1:t}), M_{l, t} \in C_b(\cX_{1:t}), M_l \text{ an $\cF^\cX$-martingale, and some positive integer } L \in \mathbb{N} \Big\}. 
\end{align*}
\citet[Proposition 2.4]{backhoff2017causal} show that the martingale $M_l$ is constructed with conditional kernels $\hat{\mu}(dx_{t+1:T}|x_{1:t})$. $M_l$ can still be assumed continuous after refining the topology on $\cX$. Therefore, the proof of the following corollary is the same as Theorem \ref{thm:cap-dual} or Corollary \ref{cor:dual}, with $\Gamma$ replaced by $\Gamma'$, and thus is omitted.
\begin{corollary}\label{cor:dual-mart}
	Under assumptions of Theorem \ref{thm:cap-dual} or Corollary \ref{cor:dual}, $D(\varepsilon; \hat{\mu})$ has another representation as 
	\begin{equation}\label{eq:dual-mart} 
		D(\varepsilon; \hat{\mu} ) = \inf_{ \lambda \geq 0, \; \gamma' \in \Gamma'} \quad  \lambda \varepsilon + \int_{\cX} F(x; \lambda, \gamma') \hat{\mu}(dx),
	\end{equation}
	with
	\begin{align*}
		F(x; \lambda, \gamma') = \sup_{y \in \cY}  \big\{& f(y) - \lambda c(x, y) + \gamma'(x, y) \big\}.
	\end{align*}
\end{corollary}

The dual problem in Theorem \ref{thm:cap-dual} or Corollary \ref{cor:dual} reveals that the worst-case scenario $\nu^*$ is the supremum over $y$ conditioning on $x$. This feature is conservative and somehow undesired when data points have abnormal spikes or observation errors. $\nu^*$ is determined pointwise and pathwise for each given $x$. The proofs of Theorem \ref{thm:cap-dual} or Corollary \ref{cor:dual} also show that the alternative measures run over all $\nu \in \cP(\cY)$. In many cases, the DM may have additional information to further restrict the selection of alternative measures. Another technical motivation is that, for a given $\nu$, the coupling set $\Pi(\hat{\mu}, \nu)$ is compact, which has been used in the proof of Theorem \ref{thm:cap-dual}. It helps to remove the boundedness assumption on the domains. Motivated by these facts, we embed structural information available into alternative measures and refer to the framework as the structural COT, or SCOT for short.

\subsection{Incorporating structural information}
Suppose structural information is available and can be modeled with a non-empty subset $\cV \subset \cP(\cY)$. Denote the primal problem with structural information as
\begin{equation}\label{eq:priori-primal}
	J(\varepsilon; \hat{\mu}, \cV) := \sup_{\nu \in \cB_\varepsilon \cap \cV} \int_\cY f(y) \nu(dy).
\end{equation}
Commonly, one can restrict $\cV$ to be a parametric space, encoded with structural information that the DM knows {\it a priori}. Then we can regard the SCOT framework as robust estimation of parametric models. Another example is to impose moment conditions on alternative measures. If we choose $\cV$ as the set of all Borel probability measures on $\cY$, it reduces to the COT framework \eqref{eq:primal}. Similarly, we have the following dual representation for \eqref{eq:priori-primal}. The assumption that $f$ is bounded from above is consistent with \citet[Theorem 2.6]{backhoff2017causal}, where the cost is bounded from below.
\begin{theorem}\label{thm:struct-dual}
	Suppose Assumption \ref{assum:lsc} holds, $\cX$ and $\cY$ are closed, $f$ is bounded from above, and $\cB_\varepsilon$ is non-empty for certain $\varepsilon > 0$. Then $J(\varepsilon; \hat{\mu}, \cV) = D(\varepsilon; \hat{\mu}, \cV)$, where
	\begin{equation}\label{eq:struct-dual}
		D(\varepsilon; \hat{\mu}, \cV) = \sup_{\nu \in \cV} \inf_{ \substack{\lambda \geq 0, \\ \gamma \in \Gamma}} \Big\{ \lambda \varepsilon + \sup_{\pi \in \Pi(\hat{\mu}, \nu)}  \int_{\cX \times \cY} \big[ f(y) - \lambda c(x, y) + \gamma(x, y) \big] \pi(dx, dy) \Big\}.
	\end{equation}
\end{theorem}

Similarly, the infimum over $\gamma \in \Gamma$ can be replaced by $\gamma' \in \Gamma'$.

One advantage of the dual representation \eqref{eq:struct-dual} is that the inner supremum over $\pi \in \Pi(\hat{\mu}, \nu)$ is closely connected with the classic Wasserstein distance. Indeed, after moving the negative sign out, we have
\begin{align}
	& \sup_{\pi \in \Pi(\hat{\mu}, \nu)}  \int_{\cX \times \cY} \big[ f(y) - \lambda c(x, y) + \gamma(x, y) \big] \pi(dx, dy) \nonumber \\
	& = - \inf_{\pi \in \Pi(\hat{\mu}, \nu)}  \int_{\cX \times \cY} \big[ - f(y) + \lambda c(x, y) - \gamma(x, y) \big] \pi(dx, dy) \nonumber \\
	& =: - \cW(\hat{\mu}, \nu; \lambda, \gamma). \label{eq:struct-dist}
\end{align}
Section \ref{sec:algo} will present several algorithms to approximate Wasserstein distance and thus the inner supremum over $\pi$.

 \subsection{Concise examples}\label{sec:motivate}
 Is there any significant difference between COT and OT primal optimizers? When do they share a common primal optimizer? Before introducing numerical algorithms for general cases, we first consider simple examples with explicit solutions to understand the distinct features between COT and OT. Besides, we show that COT and OT can share a common optimizer when the cost and the objective are separable.
 
 \noindent {\bf Example 1.} Let the time step $T = 2$ and the domain $\cX = \cY = [-1, 1]^2$. Suppose the reference measure $\hat{\mu} = 0.2 \delta_{(-1, 1)} + 0.8 \delta_{(-1, -1)}$. Equivalently, $x_1 = -1$ with probability one, $x_2 = 1$ with probability 0.2, and  $x_2 = -1$ with probability 0.8.  Suppose the cost function $c(x,y) = \mathbf{1}_{ \{ x_2 \neq y_1\} }$, which is l.s.c. Consider an objective $f((x_1, x_2)) = x_1$ that only focuses on the first time point. Then the risk under the reference measure is
 \begin{equation*}
 	\int_{\cX} f(x) \hat{\mu}(dx) = \int_{\cX} x_1 \hat{\mu}(dx)= -1.
 \end{equation*}
 Let the radius $\varepsilon = 0.2$. By Corollary \ref{cor:dual} with $\gamma (x,y) \equiv 0$, the dual problem of the OT formulation is
 \begin{align*}
 	& \inf_{\lambda \geq 0} \lambda \varepsilon + \int_\cX \sup_{y_1 \in [-1,1]} \big( y_1 - \lambda \mathbf{1}_{ \{ x_2 \neq y_1\} } \big) \hat{\mu}(dx_2) \\
 	& =  \inf_{\lambda \geq 0} \Big\{ 0.2 \lambda + 0.2 \sup_{y_1 \in [-1,1]} \big( y_1 - \lambda \mathbf{1}_{ \{y_1 \neq 1\} } \big)  + 0.8 \sup_{y_1  \in [-1,1]} \big( y_1 - \lambda \mathbf{1}_{ \{y_1 \neq -1\} } \big) \Big\} \\
 	& =  \inf_{\lambda \geq 0} \Big\{ 0.2 \lambda + 0.2 \times 1  + 0.8 \times (-1) \Big\} = -0.6.
 \end{align*}
 Thus, the worst-case measure $\nu^*$ will assign $Y_1 = X_2$, where we use capitalized letters to highlight random variables. It relies on the information from the future time, while OT does not punish this anticipative choice.
 
 For the COT problem, $\nu \in \cB_\varepsilon$ if and only if there exists a causal transport plan $\pi \in \Pi_c(\hat{\mu}, \nu)$ and $\pi(Y_1 \neq X_2) \leq 0.2$. The only information that $Y_1$ can use to predict $X_2$ is from $X_1$. However, $X_1$ does not help predict $X_2$ since $\hat{\mu}$ is separable. Then, for any causal transport plan $\pi$, we have
 \begin{align*}
 	\pi(Y_1 \neq X_2) =& \pi(Y_1 \neq 1 | X_2 =1) \hat{\mu}(X_2 = 1) + \pi(Y_1 \neq -1| X_2 = -1) \hat{\mu}(X_2 = -1) \\
 	=& 0.2 \nu(Y_1 \neq 1) + 0.8 \nu(Y_1 \neq -1) = 1 - 0.2 \nu(Y_1 = 1) - 0.8 \nu(Y_1 = -1)  \leq 0.2.
 \end{align*}
 On the other hand, $ \nu(Y_1 = 1) + \nu(Y_1 = -1) \leq 1 $. One must have $ \nu(Y_1 = -1) = 1$. Then the worst-case value is $-1$, which is smaller than the OT case. In general, the COT and OT primal values can be arbitrarily distinct. We can also obtain the same conclusion from the dual representation.

 \noindent {\bf Example 2.} To provide a counterexample that the duality may fail, consider the same setting in Example 1, but with the objective function $f((x_1, x_2)) = \infty \mathbf{1}_{\{x_1 = 1\}}$. The boundedness assumption on $f$ is violated. For the COT problem, since the causal Wasserstein ball $\cB_{\varepsilon}$ only contains measures satisfying $\nu (Y_1 = -1) = 1$, then the primal value is 0. For the dual value, note that any given $\gamma \in \Gamma$ is continuous and bounded. The supremum over $y$ will take $y_1 = 1$ since the penalty from $\gamma$ is bounded. Thus,
 \begin{align*}
 	\inf_{\lambda \geq 0, \gamma \in \Gamma} \lambda \varepsilon + \int_\cX \sup_{y_1 \in [-1,1]} \big( \infty \mathbf{1}_{\{ y_1 = 1\}} - \lambda \mathbf{1}_{ \{ x_2 \neq y_1\} } + \gamma(x, y) \big) \hat{\mu}(dx_2) = \infty.
 \end{align*}  
 This example illustrates that at least $f < \infty$ is needed for strong duality, even with a bounded domain. In contrast, both the primal and dual values for the OT problem are infinity.

Example 1 shows COT and OT have different optimizers in general. The following corollary investigates another side. When the objective and the cost are separable, COT and OT share a common optimizer.
 \begin{corollary}\label{cor:equiv}
 	Suppose the objective $f(y) = \sum^T_{t=1} f_t(y_t)$ and the cost $c(x, y) = \sum^T_{t=1} c_t(x_t, y_t)$, satisfying $c_t \geq 0$ and $c_t(x_t, x_t) = 0$. $\int_\cX |f(x)| \hat{\mu} (dx) < \infty$. Assumptions in Theorem \ref{thm:cap-dual} or Corollary \ref{cor:dual} hold. Then COT and OT share a common optimizer for $J(\varepsilon; \hat{\mu})$. More precisely, there is a measurable map $L_t: \cX_t \rightarrow \cY_t$ for each $t \in [T]$ attaining the supremum over $y$. A worst-case measure $\nu^*$ is induced by $Y_t = L_t(X_t)$. 
 \end{corollary}

Corollary \ref{cor:equiv} implies that, with the separable cost and objective, the worst-case scenario can be deduced separately and does not rely on future values. It exhausts the current state $x_t$ to find the worst-case $y_t$. For example, if we consider $f(y)$ as the sum of $y$ and $c(x, y)$ as the $l_1$ norm, then the worst-case scenario is simply the reference measure shifted by a constant depending on the radius $\varepsilon$. Thus, the temporal structure of $x$ and $y$ is the same. In this sense, we claim that COT and OT rely on the reference measure {\it pointwise and pathwise}, instead of jointly on samples.
 
\section{Sample complexity and optimization}\label{sec:complex}

The dual representation of COT \eqref{eq:minimax} and SCOT \eqref{eq:struct-dual} poses several essential difficulties to optimization. In COT, the inner maximization over $y$ is non-concave in general, which is NP-hard already. The outer minimization is over $\lambda$ and an infinite-dimensional test function space, which also appears in SCOT. To reduce the minimization over test functions to a finite-dimensional problem, we adopt neural networks to model elements from the test function space. The motivations are twofold. First, compared with the discrete scheme in \cite{acciaio2021cournot,zhou2021distri}, the parametric formulation scales well when the sample size increases. The test functions may also learn some meaningful patterns from data. Another motivation is the approximation capability of neural networks \citep{hornik1991,cybenko1989} and the modern computational tools available.  In Section \ref{sec:rad}, we prove the sample complexity of COT with tools from Rademacher complexity. Section \ref{sec:RegNN} gives an example to verify technical assumptions. Optimization algorithms are discussed later in Section \ref{sec:algo}. As preparation, we first introduce the adapted empirical measure in \cite{backhoff2020estimating} that resolves a flaw of classic empirical measure.

In this section, we further assume
\begin{assumption}\label{assum:continuous}
	$\cX_t = \cY_t = [a, b]^d$ are bounded cubes for all $t \in [T]$.  $f$ and $c(x,y)$ are continuous and satisfy Assumption \ref{assum:lsc}.
\end{assumption}
The sample complexity analysis relies on the consistency results from \cite{backhoff2020estimating} under compact interval domains. We note that \cite{acciaio2022convergence} have extended relevant results to $\R^d$ recently. However, our convergence analysis in Theorem \ref{thm:conv} still needs a compact domain to validate the universal approximation property of neural networks \cite[Theorem 1]{cybenko1989}.

\subsection{Adapted empirical measures}
The reference measure $\hat{\mu}$ is usually chosen as the empirical measure $\hat{\mu}_N$ with $N$ observed paths $x^n_{1:T}$, i.i.d. from an underlying measure $\mu$. As noted in \cite{pflug2016empirical,backhoff2020estimating}, $\cW_c (\hat{\mu}_N, \mu) \rightarrow 0$ may not hold when sample size $N \rightarrow \infty$. To understand this undesired property, consider the measure $\mu$ that has a density function with respect to the Lebesgue measure. The probability that two observed sample paths coincide at some time point is zero. Then the kernels of the empirical measure are Dirac measures almost surely; see \citet[Remark 2.1]{backhoff2020estimating}. Once we know the value of the sample path at time $t=1$, the evolution of the whole path is fully determined. There will be no penalty from $h_{l, t}$ and $g_{l, t}$. Thus, convergence in $\cW_c$ fails.   

To resolve this issue, \cite{backhoff2020estimating} propose a quantization method. By \citet[Definition 1.2]{backhoff2020estimating}, set constant $q = 1/(T+1)$ for one-dimensional path spaces ($d=1$) and $q=1/(dT)$ for $d \geq 2$. For all sample sizes $N \geq 1$, divide $[a, b]^d$ into disjoint unions of $N^{qd}$ cubes with edges of length $(b-a)/N^{q}$. Let $\varphi_N: [a, b]^d \rightarrow [a, b]^d$ map each cube to its center. Define the {\bf adapted} empirical measure as
\begin{equation*}
	\hat{\mu} ( dx_{1:T} | N, \varphi_N) := \frac{1}{N} \sum^N_{n = 1} \delta_{  \varphi_N( x^n_{1:T} )}.
\end{equation*}  
In this notation, the first $N$ means $N$ samples are used. $\varphi_N$ highlights the role of the partition map. Denote the partition formulated by $\varphi_N$ on $[a, b]^d$ as
\begin{equation*}
	\Phi_N := \{ \varphi_N^{-1} (x): x \in \varphi_N([a, b]^d) \}.
\end{equation*}
For every $1 \leq t \leq T - 1$, the product of cubes formulates
\begin{equation*}
	\Phi_{N,t} := \Big\{ \prod_{1 \leq s \leq t} B_s : B_s \in \Phi_N  \Big\}. 
\end{equation*}  
With a slight abuse of notation, we interpret $\varphi^{-1}_N(x_{1:t})$ as the product of cubes when $\varphi^{-1}_N$ is applied element-wise.  

Recall that $\mu (d x_{t+1:T} | x_{1:t})$ is the true conditional measure. With a set $A_t \in \Phi_{N,t}$, we can define the {\bf adapted} empirical conditional measure as
\begin{align*}
	\hat{\mu} (d x_{t+1:T} | A_t, N, \varphi_N) & := \frac{1}{| \{ x^n_{1:t} \in A_t, \, n \in [N] \} |} \sum_{ x^n_{1:t} \in A_t, \, n \in [N] } \delta_{  \varphi_N( x^n_{t+1:T} ) },
\end{align*}
where $|\cdot|$ is the cardinality of a set. Besides, we also need the following notation for later use:
\begin{align*}
	\hat{\mu} (d x_{t+1:T} | A_t, N, \id) & := \frac{1}{| \{ x^n_{1:t} \in A_t, \, n \in [N] \} |} \sum_{ x^n_{1:t} \in A_t, \, n \in [N] } \delta_{x^n_{t+1:T}} .
\end{align*}
Quantization maps values in the same cube to the centroid of the cube. Therefore, two sample paths have non-zero probabilities to be in the same cube and mapped to the same value at time $t=1$. Then the evolution of sample paths after $t=1$ is not fully determined. This is in contrast to the classic empirical measure in which values are almost surely distinct at $t=1$. For further details, see \citet[Figure 1 and Section 2]{backhoff2020estimating}.

\subsection{Rademacher complexity}\label{sec:rad}
In this subsection, we consider the sample complexity of the COT dual problem with the test function space $\Gamma$, which is infinite dimensional. To make the problem tractable, a general idea is to approximate $\Gamma$ with simpler sets of functions. In machine learning theory, a hypothesis set is a collection of functions that map features or attributes to a set of finite or infinite labels; see \citet[Chapter 1.4]{mohri2018foundations}. In this paper, we define a hypothesis set $\Theta_k$ as 
\begin{align}
	\Theta_k := \{&  g_{l, t}(\cdot; \theta') \text{ and } h_{l, t} (\cdot; \theta''), \text{ with } \theta', \theta'' \in \mathcal{C}_k \text{ which is compact}, \label{eq:thetak}\\
	&  l \in [L_k], \, t \in [T-1], \, \lambda \in [0, \lambda_k]\}. \nonumber
\end{align}
Here, $g_{l, t}$ and $h_{l, t}$ are parameterized by $\theta', \theta''$ from a compact subset $\mathcal{C}_k \subset \R^{m_k}$. $g_{l, t}$ and $h_{l, t}$ are jointly continuous on parameters ($\theta'$ or $\theta''$) and inputs ($x$ or $y$). We assume all functions $h_{l,t}$ and $g_{l,t}$ in the hypothesis set $\Theta_k$ are uniformly bounded by a large enough constant $C_{h, g}$, independent of sample size $N$. To approximate continuous and bounded functions, we use neural networks due to their universal approximation property \citep{hornik1991,cybenko1989}. We also restrict the Lagrange multiplier $\lambda$ in a compact set and $l$ ranging from finite integers. For simplicity, we write $\theta \in \Theta_k$ to mean a choice of parameters $(\theta', \theta'', \lambda)$. 

Denote the hypothesis set $\Theta$ with the whole test function space $\Gamma$ and $\lambda \in [0, \infty)$ as
 \begin{equation*}
 	\Theta := \{g_{l, t} \in C_b(\cX_{1:T}), \, h_{l, t} \in C_b(\cY_{1:t}), \, t \in [T-1], \, l \in [L] \text{ for some positive integer } L \in \mathbb{N}, \, \lambda \geq 0 \}.
 \end{equation*}
 With the universal approximation property of neural networks, we can approximate $\Theta$ with $\Theta_k$, when the compact set $\mathcal{C}_k$ approaches $\R^{m_k}$ and $m_k$, $C_{h, g}$, $\lambda_k$, and $L_k$ go to infinity.

Recall the dual problem of COT in \eqref{eq:minimax}. When we set $\hat{\mu}$ as the adapted empirical measure $\hat{\mu} ( dx_{1:T} | N, \varphi_N)$, the dual problem \eqref{eq:minimax} minimizes
\begin{equation*}
	\begin{aligned}
		& D(\theta, \varepsilon; \hat{\mu} \circ \varphi^{-1}_N) \\
		& := \int \sup_{y \in \cY}  \Big\{ f(y) - \lambda c(x, y) + \sum^L_{l=1} \sum^{T-1}_{t=1} h_{l,t}(y_{1:t}) \big[ g_{l, t}(x_{1:T}) \\
		& \hspace{2cm} - \int g_{l, t}( x_{1:T}) \hat{\mu} (d x_{t+1:T} | \varphi^{-1}_N (x_{1:t}), N, \varphi_N) \big] \Big\} \hat{\mu}(dx_{1:T}| N, \varphi_N) + \lambda \varepsilon \\
		& = \int \sup_{y \in \cY}  \Big\{ f(y) - \lambda c(\varphi_N(x), y) + \sum^L_{l=1} \sum^{T-1}_{t=1} h_{l,t}(y_{1:t}) \big[ g_{l,t}(\varphi_N(x_{1:T})) \\
		& \hspace{2cm} - \int g_{l,t}( \varphi_N(x_{1:T})) \hat{\mu} (d x_{t+1:T} | \varphi^{-1}_N (x_{1:t}), N, \id) \big] \Big\} \hat{\mu}(dx_{1:T}| N, \id) + \lambda \varepsilon.
	\end{aligned}
\end{equation*}
The second equality follows from the push-forward property. We use $\hat{\mu} \circ \varphi^{-1}_N$ to highlight the adapted empirical measure for abbreviation. $D(\theta, \varepsilon; \hat{\mu} \circ \varphi^{-1}_N)$ considers parameterized $g_l$ and $h_l$ in the hypothesis set $\Theta_k$ or $\Theta$, while we omit their parameters $\theta'$ and $\theta''$ for simplicity. Note that
\begin{equation*}
	\inf_{ \theta \in \Theta} D(\theta, \varepsilon; \hat{\mu} \circ \varphi^{-1}_N) = D(\varepsilon; \hat{\mu} \circ \varphi^{-1}_N),
\end{equation*}
with the right-hand side defined in \eqref{eq:minimax}. 

A function used in the middle of proofs is	
\begin{equation*}
	\begin{aligned}
		D(\theta, \varepsilon; \mu \circ \varphi^{-1}_N) := \int \sup_{y \in \cY}  \Big\{ & f(y) - \lambda c(\varphi_N(x), y) + \sum^L_{l=1} \sum^{T-1}_{t=1} h_{l,t}(y_{1:t}) \big[ g_{l,t}( \varphi_N(x_{1:T}) ) \\
		& - \int g_{l,t}( \varphi_N(x_{1:T})) \mu (d x_{t+1:T} | \varphi^{-1}_N (x_{1:t})) \big] \Big\} \mu (dx_{1:T}) + \lambda \varepsilon.
	\end{aligned}
\end{equation*}
In $D(\theta, \varepsilon; \mu \circ \varphi^{-1}_N)$, the probability measure is the true $\mu$ while the partition map remains in the integral. Besides, let $D(\theta, \varepsilon; \mu)$ be the function minimized under the true measure $\mu$ and without the partition map.

To quantify the learning capability of a function class, we recall the definition of (empirical) Rademacher complexity from \citet[Definition 3.2, Chapter 3]{mohri2018foundations}:
\begin{definition}
	The empirical Rademacher complexity of a hypothesis set $\cG$ with respect to a fixed sample $S = (x^1_{1:T}, ... , x^N_{1:T} )$ of size $N$ is defined as
	\begin{equation}\label{ERC}
		\cR_S (\cG; N) = \E_{\mathbf{\sigma}} \Big[ \sup_{g \in \cG} \frac{1}{N} \sum^N_{i=1} \sigma_i g(x^i) \Big],
	\end{equation}
where $\sigma = (\sigma_1, ..., \sigma_N)$, with $\sigma_i$s independent uniform random variables taking values in $\{ -1, + 1\}$. 

For any positive integer $N$, the Rademacher complexity of $\cG$ is the expectation of the empirical Rademacher complexity over all samples of size $N$ drawn according to $\mu$:
\begin{equation}\label{RC}
	\cR (\cG; N) = \E_{S} [ \cR_S (\cG; N) ].
\end{equation}
\end{definition}
The random variables $\sigma_i$ are called Rademacher variables. The supremum over $g \in \cG$ in \eqref{ERC} measures how well the hypothesis set $\cG$ correlates with random noise $\sigma$ over a given sample $S$. Hence, the (empirical) Rademacher complexity measures on average the richness of the function family $\cG$.

Let $\Theta_k$ be a given hypothesis set. Denote $\cR( \cG_k \circ \varphi_N; N')$ as the Rademacher complexity of the family of functions $g_{l,t}( \varphi_N(\cdot))$, with $g_{l, t}$ from $\Theta_k$ and the sample size as $N'$. $D(\theta, \varepsilon; \hat{\mu} \circ \varphi^{-1}_N)$ depends on the hypothesis set $\Theta_k$ via $F$ in \eqref{eq:minimax}. We need to consider $\cR( F( \Theta_k \circ \varphi_N); N)$ as the Rademacher complexity of the family of the function $F$ composited with $g_{l,t}( \varphi_N(\cdot))$ and $h_{l,t}( \varphi_N(\cdot))$ from $\Theta_k$. 

Our first technical result is to derive general learning guarantees for $\Theta_k$. Define a function
\begin{equation}\label{def:r}
	r(x; \delta) := 2 \cR( \cG_k \circ \varphi_N; x) + 2 C_g \sqrt{\frac{\ln(2/\delta)}{2x}}, \quad  x \in [0, \infty),
\end{equation}
for later use in the proof. $C_g > 0$ is a constant to uniformly bound all $g_{l, t}$ in $\Theta_k$ and independent of $x$. Condition 2 in Lemma \ref{lem:Rad} assumes the Rademacher complexity in \eqref{def:r} decreases when the sample size increases and the concavity of $xr(x;\delta)$ is used in Jensen's inequality. The main difficulty is that the adapted empirical measures appear in the integral with functions $g_{l, t}$.

\begin{lemma}\label{lem:Rad}
	Consider the hypothesis set $\Theta_k$ in \eqref{eq:thetak}. Suppose 
	\begin{enumerate}
		\item Assumption \ref{assum:continuous} holds; 
		\item $x r(x; \delta)$ is concave on $x \in [0, \infty)$, and $r(x; \delta)$ is decreasing on $x \in [0, \infty)$. 
	\end{enumerate}
   Then with probability at least $1 - \delta$ over the draw of an i.i.d. sample of $\mu$ of size $N$, the generalization bound is given by
	\begin{align}
		\sup_{\theta \in \Theta_k}  \left| D(\theta, \varepsilon; \mu \circ \varphi^{-1}_N) - D (\theta, \varepsilon; \hat{\mu} \circ \varphi^{-1}_N) \right|  \leq & 2 \cR( F( \Theta_k \circ \varphi_N); N)  + C \cR( \cG_k \circ \varphi_N ; N^{1 - qd(T-1)}) \nonumber \\
		& + C \sqrt{\frac{\ln[2(NL_k(T-1)+1)/\delta]}{N^{1 - qd(T-1)}}}, \label{eq:genbd}
	\end{align}
	where $L_k$ is in \eqref{eq:thetak} and $C>0$ is a constant independent of the sample size $N$.
\end{lemma}

The generalization bound \eqref{eq:genbd} provides an estimate of the largest possible difference between the empirical risk, computed using $\hat{\mu}$, and the true risk, computed using $\mu$. This bound is closely related to Talagrand's inequality for empirical processes, see \citet[Theorem 2.1]{bartlett2005local}. Although the proof technique is general, the upper bound obtained by taking the supremum over $\theta \in \Theta_k$ may be loose, since the algorithm might only select functions from smaller sets.

In the following lemma, we fix the hypothesis set $\Theta_k$ and let the sample size go to infinity. Then we can obtain the infimum of the dual problem over the hypothesis set $\Theta_k$ and under the true measure $\mu$.

\begin{lemma}\label{lem:Ninf}
	Given $\Theta_k$, suppose assumptions in Lemma \ref{lem:Rad} hold and the Rademacher complexities $\cR( F( \Theta_k \circ \varphi_N); N)$ and $  \cR( \cG_k \circ \varphi_N; N^{1 - qd(T-1)})$ in Lemma \ref{lem:Rad} converge to zero when sample size $N \rightarrow \infty$. Then   
	\begin{align*}
		\lim_{N \rightarrow \infty} \inf_{\theta \in \Theta_k} D(\theta, \varepsilon; \hat{\mu} \circ \varphi^{-1}_N) = \inf_{\theta \in \Theta_k} D(\theta, \varepsilon; \mu),
	\end{align*}
	with probability one.
\end{lemma}

Finally, let the hypothesis set be large enough to approximate $\Theta$, which follows from the universal approximation theorem of neural networks under compact domains \cite[Theorem 1]{cybenko1989}. Then we obtain the dual value with the true measure $\mu$ and $D(\varepsilon; \mu) = J(\varepsilon; \mu)$ by Corollary \ref{cor:dual}. Note that Theorem \ref{thm:conv} uses $\mu$ after the sample size has converged to infinity in Lemma \ref{lem:Ninf}. Since $\mu$ is in the causal Wasserstein ball with condition 3 in Corollary \ref{cor:dual}, then $J(\varepsilon; \mu) \geq \int_\cX f(x) \mu (dx)$, which provides an upper bound for the real risk.
\begin{theorem}\label{thm:conv}
	Suppose assumptions in Lemma \ref{lem:Ninf} hold. Moreover, for any bounded continuous function $g$, there exists a sequence of functions $g_k \in \Theta_k$ such that $g_k \rightarrow g$ pointwise when $k \rightarrow \infty$. Then
	\begin{align*}
		\lim_{k \rightarrow \infty}  \inf_{\theta \in \Theta_k} D (\theta, \varepsilon; \mu) = \inf_{\theta \in \Theta} D(\theta, \varepsilon; \mu) = D(\varepsilon; \mu) = J(\varepsilon; \mu).
	\end{align*}
\end{theorem}

In Section \ref{sec:finite} of the Appendix, we also provide a finite sample guarantee for SCOT.

\subsection{Regularized neural networks}\label{sec:RegNN}
	Indeed, many neural network architectures satisfy the assumptions imposed in  Lemma \ref{lem:Rad}, Lemma \ref{lem:Ninf}, and Theorem \ref{thm:conv}, subject to certain conditions on their parameters. A simple example is presented in this subsection.
	
	Under Assumption \ref{assum:continuous}, the input space $\cX_{1:t} = ([a, b]^d)^t$ for a generic $t \in [T]$. The family of one-layer regularized neural networks is defined as the following set of functions mapping $\cX_{1:t}$ to $\R$:
	\begin{equation}\label{eq:Ak}
	\begin{aligned}
		\cA_{k, t} := \Big\{  x_{1:t} \mapsto \sum^{n_k}_{j=1} w_j \phi(u_j \cdot x_{1:t} + \vartheta_j) \Big|\; & \sum^{n_k}_{j=1} |w_j| \leq C_{w, k}, \, \sum^{n_k}_{j=1} |\vartheta_j| \leq C_{\vartheta, k}, \\
		& \|u_j\|_2 \leq C_{u, k}, \, \forall \, j \in [n_k] \Big\},
	\end{aligned}
	\end{equation} 
	where $\phi$ is an $L_\phi$-Lipschitz function. Moreover, $\phi$ is sigmoidal, satisfying $\phi(a) \rightarrow 1$ as $a \rightarrow \infty$ and $\phi(a) \rightarrow 0$ as $a \rightarrow -\infty$; see \cite{cybenko1989}. $w_j \in \R$, $u_j \in \R^{d \times t}$, and $\vartheta_j \in \R$ are parameters of the network, each bounded by universal constants as specified in \eqref{eq:Ak}. Importantly, there exists a constant $L_\cA$, such that all functions in $\cA_{k, t}$ are $L_\cA$-Lipschitz in $x_{1:t}$. By the universal approximation theorem \cite[Theorem 1]{cybenko1989}, as the constants $n_k$, $C_{w, k}$, $C_{\vartheta, k}$, and $C_{u, k}$ tend to infinity, the family of neural networks in \eqref{eq:Ak} can uniformly approximate any continuous function within a compact domain.
	
	We set the hypothesis set $\Theta_k$ in \eqref{eq:thetak} as
	\begin{align}
		\Theta_k := \big\{& g_{l, t} \in \cA_{k, T}, \; h_{l, t} \in \cA_{k, t}, \; l \in [L_k], \; t \in [T-1], \; \lambda \in [0, \lambda_k] \big\}. \nonumber
	\end{align}
	
	\begin{lemma}\label{lem:Rad_Gk}
		Rademacher complexity of $\cG_k \circ \varphi_N$, with $\cG_k = \cA_{k, T}$, satisfies
		\begin{align}
			\cR (\cG_k \circ \varphi_N; N) \leq \frac{C_{w, k} L_\phi (C_{\vartheta, k} + C_{u, k} C_{b, a, T, d}) }{\sqrt{N}},
		\end{align}
		where $C_{b, a, T, d}$ is a constant that depends on $b, a, T$, and $d$ only.
	\end{lemma}

The following assumption is similar to \citet[Assumption 1.4]{backhoff2020estimating} and \citet[Definition 2.10]{acciaio2022convergence}. A specific example can be found in \citet[Example 1.9]{backhoff2020estimating}.
\begin{assumption}\label{assum:Lip_kernel}
	There is a version of the disintegration such that for every $1 \leq t \leq T-1$ the mapping 
	$x_{1:t} \mapsto \mu(dx_{t+1:T}|x_{1:t}) \in \cP(\cX_{t+1:T})$ is $L_\mu$-Lipschitz, where $\cP(\cX_{t+1:T})$ is endowed with the classic Wasserstein-1 distance. Moreover, $\mu$ is the law of a stochastic process $X_{1:T}$ with the following dynamic:
	\begin{equation}
		X_{t+1:T} = F_{t+1}(X_{1:t}, \varepsilon),
	\end{equation}
	for $t = 1, \ldots, T-1$ with arbitrary $X_1$. The random variable $\varepsilon$ is independent of $X_{1:t}$. Denote the law of $\varepsilon$ as $\psi$. $F_{t+1}$ is a deterministic function with suitable dimensions.
\end{assumption}

Clearly, there exist constants $C_g$, $C_h$, and $C_c$, such that $|g_{l, t}(x_{1:T})| \leq C_g$, $|h_{l, t}(y_{1:t})| \leq C_h$, and $|c(x_{1:T}, y_{1:T})| \leq C_c$. One can also choose $C_g = C_h =: C_{h, g}$, but we keep them different to make the result more clear. The following lemma gives an estimate of $\cR( F( \Theta_k \circ \varphi_N); N)$.
\begin{lemma}\label{lem:Rad_F}
	Suppose Assumptions \ref{assum:continuous} and \ref{assum:Lip_kernel} hold. Then
	\begin{equation}
		\begin{aligned}
			& \cR( F( \Theta_k \circ \varphi_N); N) \\
			& \leq \frac{C_c \lambda_k + 2L_k (T-1)(C_g + C_h) C_{w, k} L_\phi (C_{\vartheta, k} + C_{u, k} C_{b, a, T, d}) }{\sqrt{N}} \\
			& \quad + \frac{ L_k(T-1)C_h(L_\cA L_\mu + L_\cA) C_{b, a, T, d}}{N^q},
		\end{aligned}
	\end{equation}
where $C_{b, a, T, d}$ is a constant that depends on $b, a, T$, and $d$ only, but can be different with the one in Lemma \ref{lem:Rad_Gk}.
\end{lemma}

\begin{theorem}\label{thm:Reg}
	Suppose Assumptions \ref{assum:continuous} and \ref{assum:Lip_kernel} hold. The constants $n_k$, $C_{w, k}$, $C_{\vartheta, k}$, $C_{u, k}$, $L_k$, and $\lambda_k$ approach infinity when $k \rightarrow \infty$. Then the claims in Lemma \ref{lem:Rad}, \ref{lem:Ninf} and Theorem \ref{thm:conv} hold.
\end{theorem}

From the proof of our previous results, particularly Lemma \ref{lem:Rad_F}, we observe a distinct advantage of Rademacher complexity. Compared with other measures of complexity, such as covering numbers and Vapnik-Chervonenkis (VC) dimension, Rademacher complexity effectively leverages the structure inherent in the function class. This includes properties like the Lipschitz continuity of activation functions and the linear layers within neural networks \cite[Theorem 12]{bartlett2002rademacher}. Hence, the proof yields explicit estimates of Rademacher complexity. Besides, Rademacher complexity is data-dependent, allowing for numerical estimation methods even when explicit estimates are unavailable \cite[Theorem 11]{bartlett2002rademacher}. \cite{klkesk2012comparison} noted that Rademacher complexity provides better bounds across all experimental conditions in classification problems involving polynomials.

\subsection{Optimization algorithms}\label{sec:algo}

%%%%%%%%% Optimize COT %%%%%%%%%%%%%
In this subsection, we present the numerical algorithms for the dual problems \eqref{eq:minimax} [or \eqref{eq:dual-mart}] and \eqref{eq:struct-dual}, which are closely related to minimax optimization. Suppose the objective $f$ and the cost $c(x,y)$ are differentiable, which allows us to apply gradient descent ascent (GDA) methods \citep{jin2020local,lin2020gradient}.

First, consider the optimization algorithm for the dual problem of COT and OT. It is easier to implement the test function space $\Gamma'$ with martingales. We model $h_{l, t}$ and $M_{l, t}$ using neural networks, as described in more detail in the Appendix. To impose the martingale condition on $M_{l, t}$, we adopt the same method in \cite{xu2020cot} and penalize $M_{l, t}$ by the following term:
\begin{equation*}
	p(M; \varphi_N) := \frac{1}{NT} \sum^L_{l = 1} \sum^{T-1}_{t=1} \Big| \sum^N_{n=1} \frac{M_{l, t+1}(\varphi_N(x^n_{1:t+1})) - M_{l,t}(\varphi_N(x^n_{1:t})) }{\sqrt{\text{Var}[M_l]} + \eta} \Big|,
\end{equation*}
where $N$ is the sample size of $\hat{\mu}$, $\text{Var}[M_l]$ is the empirical variance of $M_l$, and $\eta > 0$ is a small constant to avoid dividing by zero. This term penalizes $M_{l, t}$ if it deviates from the martingale property when the changes at every time step are non-zero on average.  

In summary, the COT optimization problem with martingale penalization is
\begin{align}\label{eq:COTopt}
	 \inf_{ \lambda \geq 0, \; \gamma' \in \Gamma'} \lambda \varepsilon + \int \sup_{y \in \cY}  \big\{ f(y) - \lambda c(x, y) + \gamma'(x, y) \big\} \hat{\mu}(dx_{1:T}| N, \varphi_N) + \xi p(M; \varphi_N),
\end{align}
and $\xi > 0$ is a weighting constant. Note that the minimization is over $\lambda$ and parameters of $\gamma'$, which is a finite-dimensional problem. We propose Algorithm \ref{algo:cot} to evaluate inner maximization and outer minimization alternatively, which is a special case of the widely used GDA algorithm \citep{jin2020local,lin2020gradient}. Experimentally, multiple inner maximization steps yield more stable results, while the computational burden is higher.
\begin{algorithm}
	\begin{algorithmic}[1]
		\STATE {\bf Input:} Objective $f$, network $\gamma'(x, y)$, initial $\lambda > 0$
		\FOR{$O$ steps}
		\STATE Sample $x$ and construct the adapted empirical measure $\hat{\mu}(dx_{1:T}| N, \varphi_N)$
		\FOR{$I$ steps}
		\STATE Perform a gradient ascent step on \eqref{eq:COTopt} over $y$ and project $y$ into $\cY$
		\ENDFOR
		\STATE Perform a gradient descent step on \eqref{eq:COTopt}  over $\lambda$ and parameters of $\gamma'$. Truncate parameters of $\gamma'$
		\ENDFOR
		\STATE {\bf Output:} network $\gamma'$, $y$, constant $\lambda$, and dual value
	\end{algorithmic}
	\caption{}
	\label{algo:cot}
\end{algorithm}

%%%%%%%%%% Optimize SCOT %%%%%%%%%%%%%

For SCOT, suppose $\nu \in \cV$ is a parameterized model, which we refer to as a generator of alternative measures. Unlike in the COT case, where we take the supremum over $y$ in \eqref{eq:COTopt}, the SCOT dual problem \eqref{eq:struct-dual} requires evaluating the Wasserstein distance, which can be expensive to compute exactly. Various optimization methods can reduce the computational burden and approximate the Wasserstein distance $\cW(\hat{\mu}, \nu; \lambda, \gamma)$ in \eqref{eq:struct-dist}.

One popular approach, known as Sinkhorn's algorithm \citep{cuturi2013sinkhorn,genevay2016stochastic,xu2020cot}, is used when both $\hat{\mu}$ and $\nu$ are discrete measures. For a transport plan $\pi$, denote the entropy as $S(\pi)= - \sum_{i,j} \pi_{ij} \ln \pi_{ij}$ for the discrete case with $0 \ln 0 = 0$ and $S(\pi)= - \int_{\cX \times \cY} \ln \pi(dx, dy) \pi(dx, dy)$ for the continuous case. \cite{cuturi2013sinkhorn} proposes to regularize the optimal transport problem with the entropy term and consider
\begin{equation*}
	\inf_{ \pi \in \Pi(\hat{\mu}, \nu)} \Big\{ \int_{\cX \times \cY} \big[ - f(y) + \lambda c(x, y) - \gamma(x, y) \big] \pi(dx, dy)  - \lambda \tau S(\pi)  \Big\}
\end{equation*}
with a parameter $\tau > 0$. The entropy-regularized optimal transport problem can be efficiently solved using Sinkhorn's algorithm. Denote the corresponding minimizer as $\pi^*(dx, dy;\tau)$ and the corresponding regularized distance as
\begin{equation}\label{eq:sink}
	\cW(\hat{\mu}, \nu; \lambda, \gamma, \tau) := \int_{\cX \times \cY} \big[ - f(y) + \lambda c(x, y) - \gamma(x, y) \big] \pi^*(dx, dy; \tau). 
\end{equation}
Finally, with the martingale regularization term on $M$, the SCOT optimization problem becomes
\begin{equation}\label{eq:reg_obj}
	\sup_{\nu \in \cV} \inf_{ \lambda \geq 0, \; \gamma' \in \Gamma'} \lambda \varepsilon - \cW(\hat{\mu}, \nu; \lambda, \gamma', \tau) + \xi p(M; \varphi_N).
\end{equation}

Algorithm \ref{algo:SCOT} presents the optimization method for \eqref{eq:reg_obj}. The only difference between Algorithms \ref{algo:cot} and \ref{algo:SCOT} is the maximization over $y$ is replaced with the Wasserstein distance calculated by Sinkhorn's algorithm.

\begin{algorithm}
	\begin{algorithmic}[1]
			\STATE {\bf Input:} Objective $f$, network $\gamma'(x, y)$, generator $\nu$, initial $\lambda > 0$
			\FOR{$I$ steps}
			\STATE Sample $x$ and construct the adapted empirical measure
			\STATE Simulate sample $y \sim \nu$ and construct the adapted empirical measure
			\STATE Calculate \eqref{eq:sink} with Sinkhorn's algorithm
			\STATE Perform a gradient descent step on \eqref{eq:reg_obj} over $\lambda$ and $\gamma'$ and truncate parameters
			\STATE Perform a gradient ascent step on \eqref{eq:reg_obj} over $\nu$ and truncate parameters 
			\ENDFOR
			\STATE {\bf Output:} network $\gamma'$, generator $\nu$, constant $\lambda$, and dual value
		\end{algorithmic}
	\caption{}
	\label{algo:SCOT}
\end{algorithm}

If $\nu$ is continuous and $\hat{\mu}$ is discrete, an average stochastic gradient descent algorithm proposed by \cite{genevay2016stochastic} can be used to replace Sinkhorn's algorithm in Algorithm \ref{algo:SCOT}.

The optimization problems \eqref{eq:COTopt} and \eqref{eq:reg_obj} are non-convex and non-concave in the minimization and maximization steps. Convergence analysis is still an open problem, but it has received great attention in recent years \citep{jin2020local,lin2020gradient}. Experimentally, we find that the algorithms converge to some stationary points when appropriate learning rates and initial values are chosen.

\section{Numerical analysis}\label{sec:num}

Portfolio selection is a fundamental problem in mathematical finance. A straightforward diversification rule allocates a fraction $1/d$ of total wealth to each of the $d$ assets available for investment at each rebalancing date. We refer to this strategy as the naive strategy. \cite{demiguel2009} discovered that many advanced strategies, such as the sample-based mean-variance model, fail to outperform the naive portfolio in terms of Sharpe ratio during out-of-sample testing. \cite{pflug2012naive} demonstrated that in a single-period scenario, the naive strategy is optimal when an investor is faced with a sufficiently high degree of model uncertainty quantified by the Wasserstein distance, i.e. when the Wasserstein ball radius $\varepsilon \rightarrow \infty$.

In subsequent sections, we adopt the naive strategy as our benchmark and investigate whether a distributionally robust mean-variance portfolio using our SCOT method can yield superior performance. The code is publicly available at \url{https://github.com/hanbingyan/SCOT_examples}.

\subsection{Worst-case returns estimation of the naive strategy}

First, we aim to estimate the worst out-of-sample performance of the naive strategy. We select five stocks with tickers 'MMM', 'MSFT', 'JPM', 'AMZN', and 'XOM'. The weekly price data span a five-year period from July 15, 2019, to July 15, 2024, downloaded from Yahoo Finance. Treasury rates serve as the risk-free interest rate, obtained from the U.S. Department of the Treasury website. We denote the random vector $X_t$ as the weekly return of the five stocks during week $t$. Consider a four-week period by setting $T=4$. Denote the reference measure of $X_{1:4}$ as $\hat{\mu}$. Assuming weekly rebalancing for the naive strategy, the non-robust estimation of the 4-week expected return is expressed as
\begin{equation}\label{eq:nonrob_rtn}
	\int_{\cX} \prod^4_{t=1} \left(1 + \frac{\sum^5_{i=1} x_{t, i}}{5}\right) \hat{\mu}(d x) - 1.
\end{equation}

In practice, the reference measure $\hat{\mu}$ represents the weekly historical returns. However, due to the finite sample size and nonstationarity, the non-robust estimate \eqref{eq:nonrob_rtn} may not accurately represent the expected return in the next four-week period. Consequently, when the investor is averse to model uncertainty, she seeks to estimate the worst return of the naive strategy. For the OT and SCOT frameworks, we adopt the $L_1$ distance $c(x, y) = |x - y|$. The worst return using the OT method is given by
\begin{equation}\label{eq:OT_rtn}
	\inf_{\nu: \cW(\hat{\mu}, \nu) \leq \varepsilon} \int_\cY \prod^4_{t=1} \left(1 + \frac{\sum^5_{i=1} y_{t, i}}{5}\right) \nu(dy) - 1.
\end{equation}

In contrast, the SCOT method considers
\begin{equation}\label{eq:SCOT_rtn}
	\inf_{\nu \in \cB_\varepsilon \cap \cV} \int_\cY \prod^4_{t=1} \left(1 + \frac{\sum^5_{i=1} y_{t, i}}{5}\right) \nu(dy) - 1.
\end{equation}

\begin{figure}[H]
	\centering
	\begin{minipage}{0.48\textwidth}
		\centering
		\includegraphics[width=0.95\textwidth]{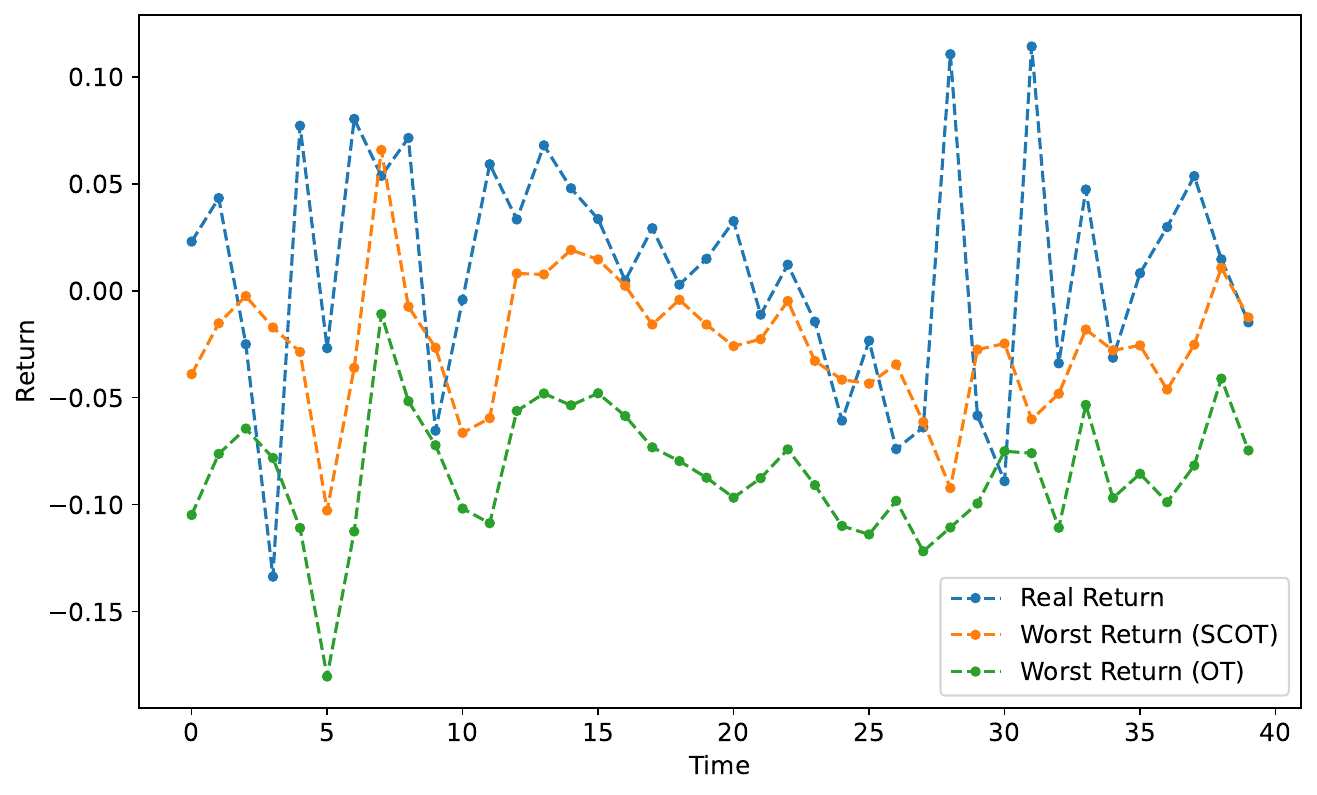}
		\subcaption{Radius $\varepsilon=0.05$}\label{fig:005}
	\end{minipage}
	\begin{minipage}{0.48\textwidth}
		\centering
		\includegraphics[width=0.95\textwidth]{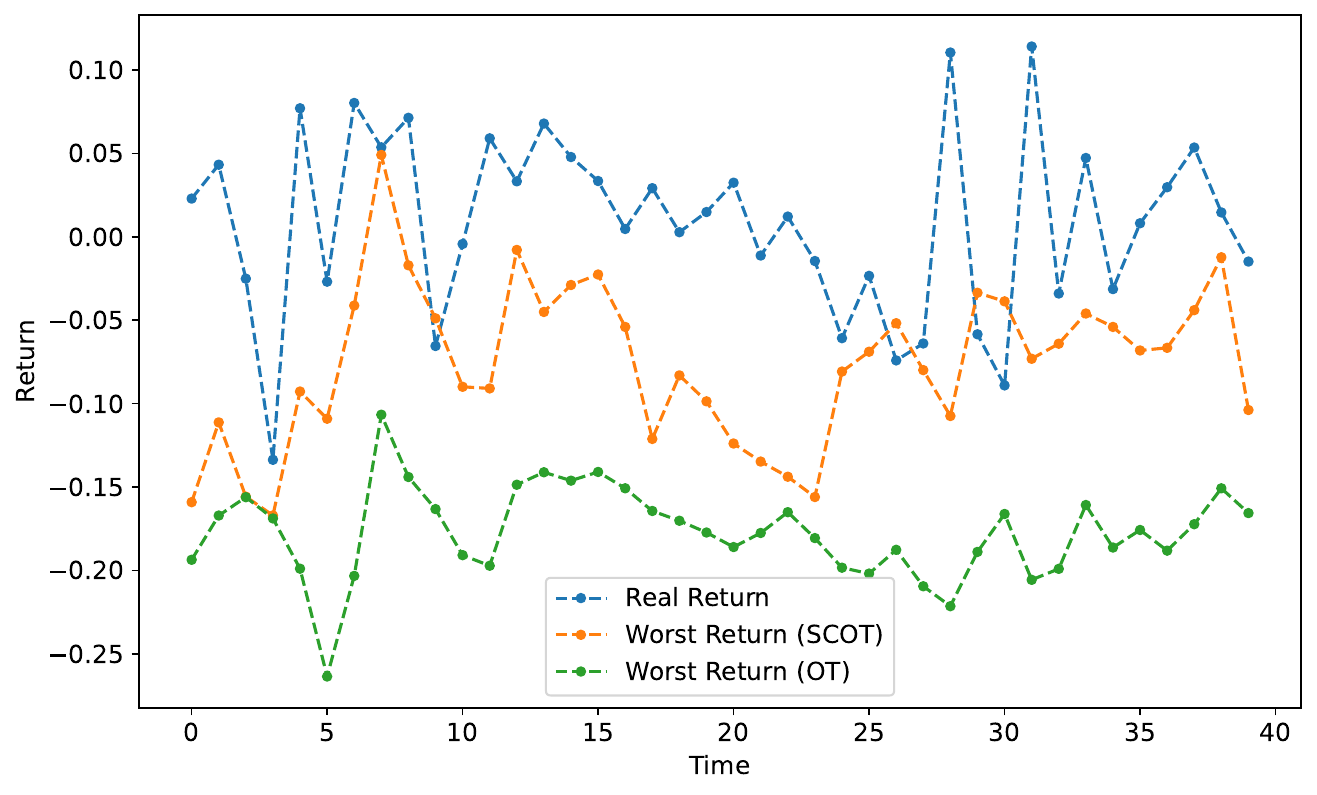}
		\subcaption{Radius $\varepsilon=0.1$}\label{fig:01}
	\end{minipage}%
	\begin{minipage}{0.48\textwidth}
		\centering
		\includegraphics[width=0.9\textwidth]{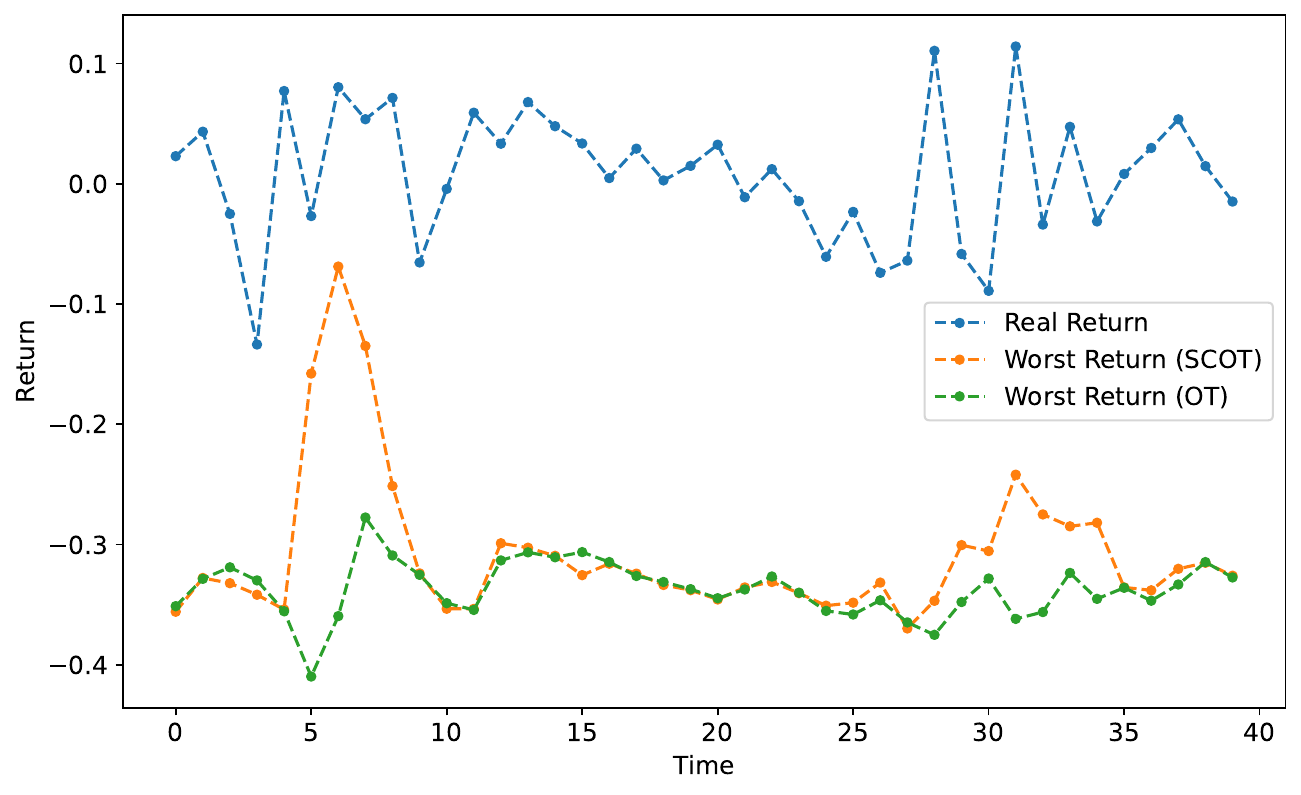}
		\subcaption{Radius $\varepsilon=0.2$}\label{fig:02}
	\end{minipage}%
	\caption{Worst-case returns estimation of the naive strategy}\label{fig:RtnEst}
\end{figure}

We define $\cV$ as a set of alternative distributions represented by Long Short-Term Memory (LSTM) networks. In contrast to the OT method \eqref{eq:OT}, which identifies the worst-case $\nu^*$ pointwise and pathwise for each given $x$, LSTM is less conservative and preserves the temporal structure. Figure \ref{fig:RtnEst} illustrates real out-of-sample returns from the naive strategy alongside worst-case returns estimated by OT and SCOT methods. Under different values of the radius $\varepsilon$, the OT method significantly underestimates actual returns. In contrast, the SCOT method is much less conservative, while it occasionally overestimates the real returns. Under high model uncertainty with $\varepsilon = 0.2$, Figure \eqref{fig:02} demonstrates that SCOT predictions align more closely with OT predictions, yet SCOT continues to provide tighter worst-case estimates at several points. Additionally, some SCOT predictions in Figure \eqref{fig:02} appear slightly lower than their OT counterparts, which seems to contradict the theoretical result stating that \eqref{eq:SCOT_rtn} should not be smaller than \eqref{eq:OT_rtn}. This discrepancy arises from estimation errors, as \eqref{eq:SCOT_rtn} is computed using only a finite sample from $\nu$.

Based on these findings from Figure \ref{fig:RtnEst}, it is anticipated that the OT method is overly conservative and may overlook certain investment opportunities, a hypothesis explored further in the subsequent subsection.

\subsection{Distributionally robust mean-variance portfolios with SCOT}

Suppose the initial wealth of the investor is \$1. For simplicity, we assume the investor uses the same investment weight vector $\alpha$ for weekly rebalancing. The terminal wealth $W_N(\alpha)$ after $N$ weeks is given by
\begin{equation}
	W_N(\alpha) = \prod^N_{t = 1} \left( 1 + \alpha^\top x_t \right),
\end{equation}
where $x_t$ denotes the returns of five assets in week $t$. Mean-variance portfolio selection determines the optimal investment weight by considering
\begin{equation}\label{eq:MV}
	\inf_{\alpha: \alpha^\top \mathbf{1} = 1} \text{Var}_{\hat{\mu}} [W_N(\alpha)] - \zeta \E_{\hat{\mu}}[W_N(\alpha)],
\end{equation}
where $(x_1, \ldots , x_N) \sim \hat{\mu}$ represents the empirical measure from historical data, and $\zeta \geq 0$ indicates the level of risk tolerance. To reformulate problem \eqref{eq:MV} in a linear form with respect to $\hat{\mu}$, we introduce an auxiliary variable $a \in \R$ as the mean value of terminal wealth. This allows us to rewrite the optimization problem \eqref{eq:MV} as
\begin{equation}\label{MV_recast}
	\inf_a \inf_{\alpha: \alpha^\top \mathbf{1} = 1} \E_{\hat{\mu}}\left[ (W_N(\alpha) - a)^2  - \zeta W_N(\alpha)  \right].
\end{equation}

Since $\hat{\mu}$ deviates from the true measure $\mu$, we enhance robustness by employing the OT and SCOT methods. The OT method examines
\begin{equation}\label{OT_MV}
	\inf_a \inf_{\alpha: \alpha^\top \mathbf{1} = 1} \sup_{\nu: \cW(\hat{\mu}, \nu) \leq \varepsilon} \E_{\nu}\left[ (W_N(\alpha) - a)^2  - \zeta W_N(\alpha) \right],
\end{equation}
where $\cW(\hat{\mu}, \nu)$ is defined in \eqref{eq:OT-dist} using the $L_1$ distance.

The SCOT method considers 
\begin{equation}\label{SCOT_MV}
	\inf_a \inf_{\alpha: \alpha^\top \mathbf{1} = 1} \sup_{\nu \in \cB_\varepsilon \cap \cV} \E_{\nu}\left[ (W_N(\alpha) - a)^2  - \zeta W_N(\alpha)  \right],
\end{equation}
where the causal Wasserstein ball $\cB_\varepsilon$ is also defined with the $L_1$ distance, and $\cV$ is modeled by the LSTM networks. The inner maximization problem in \eqref{OT_MV} and \eqref{SCOT_MV} can be addressed using the algorithms detailed in Section \ref{sec:algo} for each given pair $(a, \alpha)$. We then minimize over $(a, \alpha)$ in the outer loop and iterate between outer minimization and inner maximization. Experimental results demonstrate convergence of the algorithms.

Following \cite{demiguel2009}, we adopt the Sharpe ratio as the criterion to evaluate strategy performance. Table \ref{tab:Naive} reports the mean and standard deviation (STD) of the monthly excess return (portfolio return minus the risk-free rate) achieved by the naive strategy over a 20-month period. As a benchmark, the naive strategy achieves an annualized Sharpe ratio of 0.6633. 

Another benchmark is the non-robust optimal strategy obtained from \eqref{MV_recast}. Each month, we determine the optimal weight by solving \eqref{MV_recast} using historical data from the most recent 10 months. The results in Table \ref{tab:Nonrobust}, based on out-of-sample tests spanning the same 20 months as Table \ref{tab:Naive}, indicate that the non-robust strategy exhibits higher mean returns and STD compared to the naive strategy, ultimately resulting in higher Sharpe ratios across various levels of risk tolerance.
 
\begin{table}[H]
	\small
	\centering
	\begin{tabular}{ccc}
		\hline
		Excess Return Mean	& Excess Return STD	& Annualized Sharpe Ratio \\
		\hline
		0.01	&   0.052	 &   0.6633 \\
		\hline
	\end{tabular}
	\caption{Return statistics of the naive strategy}\label{tab:Naive}
\end{table}

\begin{table}[H]
	\small
	\centering
	\begin{tabular}{cccc}
		\hline
		Risk tolerance $\zeta$	& Excess Return Mean	& Excess Return STD	& Annualized Sharpe Ratio \\
		\hline
		0.0	& 0.0199 &	0.0624	& 1.1033 \\
		0.1	& 0.0242 &	0.0610	& 1.3752 \\
		1.0	& 0.0224 &	0.0582 &	1.3342 \\
		\hline
	\end{tabular}
	\caption{Return statistics of the non-robust strategy}\label{tab:Nonrobust}
\end{table}

Tables \ref{tab:OT} and \ref{tab:SCOT} present a comparison of OT and SCOT strategies across different combinations of risk tolerance $\zeta$ and radius $\varepsilon$. In eight of the nine cases examined, SCOT strategies yield a higher Sharpe ratio than OT strategies. Moreover, SCOT strategies outperform both non-robust and naive approaches for smaller radii. Notably, both OT and SCOT strategies converge towards the naive strategy as the radius is large ($\varepsilon = 0.2$), with this effect being more pronounced for higher risk tolerance levels (larger $\zeta$), i.e. the investor is less risk averse. However, a theoretical justification for this convergence to the naive strategy in a multi-period setting, along with an analysis of convergence rates under varying risk tolerance levels, remains an open question.

As a side note, applying the OT formulation in \eqref{OT_MV} with the same class $\cV$ improves performance relative to the original \eqref{OT_MV}, but still yields lower average Sharpe ratios compared to the SCOT method in \eqref{SCOT_MV}. For reproducibility, detailed instructions are provided in our code repository and omitted here for brevity. Since our problem is multi-period, the causality constraint is imposed. In contrast, in a single-period setting, \citet[Section 4.2]{blanchet2022structural} incorporates structural information into the cost function via implied volatility, and no causality condition is required.

\begin{table}[H]
	\small
	\centering
	\begin{tabular}{ccccc}
		\hline
		Risk tolerance $\zeta$ & Radius $\varepsilon$ & Excess Return Mean	& Excess Return STD	& Annualized Sharpe Ratio \\
		\hline
		0.0	 & 0.05	& 0.0102 &	0.0523	& 0.6751 \\
		0.0	& 0.10	& 0.0117 &	0.0540	& 0.7486 \\
		0.0	& 0.20 & 0.0177	&   0.0523  & 1.1725 \\
		0.1	& 0.05 & 0.0121	&   0.0512  & 0.8189 \\
		0.1	& 0.10 & 0.0119	&   0.0532 &  0.7778 \\
		0.1	& 0.20 & 0.0205	&   0.0495 &  1.4362 \\
		1.0	& 0.05 & 0.0099	&   0.0518 &  0.6650 \\
		1.0	& 0.10 & 0.0101	&   0.0523 &  0.6688 \\
		1.0	& 0.20 & 0.0116	&   0.0522 &  0.7708 \\
		\hline
	\end{tabular}
	\caption{Return statistics of the OT strategy}\label{tab:OT}
\end{table}

\begin{table}[H]
	\small
	\centering
	\begin{tabular}{ccccc}
			\hline
			Risk tolerance $\zeta$ & Radius $\varepsilon$ & Excess Return Mean	& Excess Return STD	& Annualized Sharpe Ratio \\
			\hline
			0.0	&  0.05	&  0.0208	&  0.0561	& 1.2829 \\
			0.0	&  0.10	& 0.0214	& 0.0559	& 1.3269  \\
			0.0 &  0.20	& 0.0192	& 0.0566	& 1.1750 \\
			0.1	&  0.05	& 0.0351	& 0.0484	& 2.5165 \\
			0.1	&  0.10	& 0.0295	& 0.0492	& 2.0754 \\
			0.1	&  0.20	& 0.0104	& 0.0560	& 0.6449 \\
			1.0	&  0.05	& 0.0263	& 0.0499	& 1.8234 \\
			1.0	&  0.10	& 0.0177	& 0.0531	& 1.1564 \\
			1.0	&  0.20	& 0.0160	& 0.0573	& 0.9652 \\
			\hline
		\end{tabular}
	\caption{Return statistics of the SCOT strategy}\label{tab:SCOT}
\end{table}

\section{Conclusion}\label{sec:conclude}

This work incorporates a causality constraint and structural information in distributionally robust risk evaluation. There are several open questions remain to be addressed. First, after approximating the dual problems of SCOT and COT with neural networks, these problems belong to the general non-convex non-concave minimax optimization, and convergence analysis is less understood. Second, the quantization method in \cite{backhoff2020estimating} suffers from the curse of dimensionality, which may be impractical for high-dimensional data. Future work could explore solutions to these problems.

\section*{Acknowledgment}
The author expresses gratitude to the anonymous referees and editors for their valuable comments and suggestions that have greatly improved this manuscript. Additionally, the author thanks seminar participants at Nanyang Technological University, University of Michigan, The Hong Kong University of Science and Technology (Guangzhou), and The Chinese University of Hong Kong (Shenzhen) for their helpful comments during presentations of this work. This work was partially conducted when the author was a postdoctoral researcher in the Department of Mathematics at the University of Michigan. He expresses gratitude to the University of Michigan for providing support and an atmosphere conducive to this work.

\section*{Conflicts of interest}
The author declares no conflicts of interest.

\section*{Data availability statement}
The data can be downloaded from Yahoo Finance at \url{https://finance.yahoo.com} and the U.S. Department of the Treasury website at \url{https://home.treasury.gov/policy-issues/financing-the-government/interest-rate-statistics}. Our code is publicly available at \url{https://github.com/hanbingyan/SCOT_examples}.

%\bibliographystyle{apalike}
%\bibliography{ref.bib}

\appendix

\section{Proofs of Section \ref{sec:dual}}
 
We recall Sion's minimax theorem from \cite{sion1958general}. It is sufficient to require one of $X$ and $Y$ is compact, see \citet[Corollary 3.3]{sion1958general}.
	\begin{theorem}[Sion's minimax theorem]
		Let $X$ be a compact convex subset of a linear topological space and $Y$ a convex subset of a linear topological space. If $f$ is a real-valued function on $X \times Y$ with
		\begin{itemize}
			\item [(1)] $f(x, \cdot)$ upper semicontinuous and quasi-concave on $Y$, $\forall x \in X$, and
			\item[(2)] $f(\cdot, y)$ lower semicontinuous and quasi-convex on $X$, $\forall y \in Y$,
		\end{itemize}
		then
		\begin{equation*}
			\min_{x \in X} \sup_{y \in Y} f(x, y) = \sup_{y \in Y} \min_{x \in X} f(x, y).
		\end{equation*}
	\end{theorem}

\begin{proof}[Proof of Theorem \ref{thm:cap-dual}]
	With slightly abuse of notations, we rewrite $D(\varepsilon; \hat{\mu}) = D(f)$ to highlight the dependence on $f$. For simplicity, let $U_b$ and $C_b$ be the sets of bounded u.s.c. and bounded continuous functions on $\cY$, respectively. $\cM(\cY)$ is the set of all countably additive, finite, positive Borel measures on $\cY$.
	
	If we can prove that
	\begin{enumerate}
		\item $D(f)$ is increasing and convex in $f$ and $D(f) < \infty$ for bounded $f$;
		\item $D(f_n) \downarrow D(0)$ for any sequence $\{ f_n \}$ in $C_b$ that $f_n \downarrow 0$ pointwise;
		\item $D^*_{C_b} (\nu) := \sup_{f \in C_b} \left[ \int_\cY f(y) \nu(dy) - D(f) \right] =  \sup_{f \in U_b} \left[ \int_\cY f(y) \nu(dy) - D(f) \right] =: D^*_{U_b} (\nu)$, for every $\nu \in \cM(\cY)$;
		\item 
		\begin{align}\label{eq:Dc}
			D^*_{C_b} (\nu) = \left\{\begin{array}{rll}
				& \inf_{ \pi \in \Pi(\hat{\mu}, \nu)}   \sup_{ \lambda \geq 0, \; \gamma \in \Gamma} \left\{ \lambda \big( \int c(x, y) d\pi - \varepsilon \big) - \int \gamma(x, y) d\pi \right\}, \text{ $\nu \in \cP(\cY)$}, \\
				& \infty, \text{ $\nu \in \cM(\cY)$ is not a probability.}
			\end{array}\right.
		\end{align}  
	\end{enumerate}
	Then by \citet[Theorem 2.2 and Proposition 2.3]{bartl2019robust}, for every $f \in C_b$, $D(f) = \max_{\nu \in \cM(\cY)} \left[ \int f(y) \nu(dy) - D^*_{C_b} (\nu) \right]$, which is exactly $J(\varepsilon; \hat{\mu})$. The maximum indicates the existence of a primal optimizer.
	
	\noindent {\bf Proof of Property 1.} Consider two functions $f_1 \leq f_2$, for given $\lambda$ and $\gamma$, one has
	\begin{equation*}
		\sup_{y \in \cY}  \big\{ f_1(y) - \lambda c(x, y) + \gamma(x, y) \big\} \leq 	\sup_{y \in \cY}  \big\{ f_2(y) - \lambda c(x, y) + \gamma(x, y) \big\},
	\end{equation*}
	which implies $D(f_1) \leq D(f_2)$. 
	
	For convexity, let $ t \in [0, 1]$ and consider two functions $f_1$ and $f_2$. By definition,
	\begin{align*}
		& D(t f_1 + (1-t) f_2) \\
		& = \inf_{ \lambda \geq 0, \gamma \in \Gamma} \lambda \varepsilon + \int_{\cX} \sup_{y \in \cY} \big\{ t f_1(y) + (1-t) f_2(y) - \lambda c(x, y) + \gamma(x, y) \big\} \hat{\mu}(dx) \\
		& \leq \inf_{ \substack{\lambda_1 \geq 0, \gamma_1 \in \Gamma \\ \lambda_2 \geq 0, \gamma_2 \in \Gamma}} (t \lambda_1 + (1-t) \lambda_2) \varepsilon \\ 
		& \quad + \int_{\cX} \sup_{y \in \cY} \big\{ t f_1(y) + (1-t) f_2(y) -(t \lambda_1 + (1-t) \lambda_2) c(x, y) + t \gamma_1(x, y) + (1-t) \gamma_2(x,y) \big\} \hat{\mu}(dx) \\
		& \leq \inf_{ \substack{\lambda_1 \geq 0, \gamma_1 \in \Gamma \\ \lambda_2 \geq 0, \gamma_2 \in \Gamma}} (t \lambda_1 + (1-t) \lambda_2) \varepsilon + \int_{\cX} \sup_{y \in \cY} \big\{ t f_1(y) - t \lambda_1 c(x, y) + t \gamma_1(x, y) \big\} \hat{\mu}(dx) \\ 
		& \quad + \int_{\cX} \sup_{y \in \cY} \big\{ (1-t) f_2(y) - (1-t) \lambda_2 c(x, y) + (1-t) \gamma_2(x,y) \big\} \hat{\mu}(dx) \\ 
		& = t D(f_1) + (1-t) D(f_2).
	\end{align*}
	
	To show $D(f) < \infty$ for bounded $f$, consider a constant $m$. Since $c(x, y)$ is non-negative, by choosing $\gamma = 0$, we obtain
	\begin{align*}
		D(m) \leq  \inf_{ \lambda \geq 0 } \lambda \varepsilon + \int_{\cX} \sup_{y \in \cY} \big\{ m - \lambda c(x, y) \big\} \hat{\mu}(dx) \leq \inf_{ \lambda \geq 0 } \lambda \varepsilon + m = m.
	\end{align*}
	Moreover, replacing the inner maximum with $y = x$ and noting $\int_\cX \gamma(x, x) \hat{\mu}(dx) = 0$, 
	\begin{align*}
		D(m) & \geq \inf_{ \lambda \geq 0, \; \gamma \in \Gamma } \lambda \varepsilon + \int_{\cX} \big\{ m - \lambda c(x, x) + \gamma(x, x) \big\} \hat{\mu}(dx) \\
		& = m + \inf_{\lambda \geq 0} \lambda \varepsilon - \lambda \int_\cX c(x, x) \hat{\mu}(dx) = m.
	\end{align*} 
	The last equality follows from the assumption that $\int_\cX c(x, x) \hat{\mu}(dx) \leq \varepsilon$. Thus, $D(m) = m$. Together with monotonicity, we have $D(f) < \infty$ for bounded $f$.
	
	\noindent {\bf Proof of Property 2.} Consider a sequence $\{ f_n \}$ in $C_b$ decreasing to 0 pointwise. Since $f_1$ is bounded, there exists a constant $m$ such that $f_1 \leq m$. Fix an arbitrary constant $\delta > 0$ and choose $\lambda>0$ such that $\lambda \varepsilon \leq \delta$. 
	
	For a given $x$, 
		\begin{align*}
			\sup_{y \in \cY} \big\{ f_n(y) - \lambda c(x, y) \big\} \geq \sup_{y \in \cY} \big\{0 - \lambda c(x, y) \big\} = - \lambda \inf_{y \in \cY} c(x, y) > - \infty.
		\end{align*}
		Then we consider a maximizing sequence $\{ y_i \}_{i=1}^\infty$ such that
		\begin{align*}
			\lim_{i \rightarrow \infty} \left( f_n(y_i) - \lambda c(x, y_i) \right) = \sup_{y \in \cY} \big\{ f_n(y) - \lambda c(x, y) \big\}.
		\end{align*}
		$\{ y_i \}_{i=1}^\infty$ should be in a ball $B_r(x)$ centered at $x$ with a finite radius $r \geq 0$. Otherwise, suppose there is a subsequence with $|x - y_{i_j}| \rightarrow \infty$. By condition 3, we obtain $c(x, y_{i_j}) \rightarrow \infty$. It contradicts with the fact that $\sup_{y \in \cY} \big\{ f_n(y) - \lambda c(x, y) \big\} > - \infty$.

	Thus, 
	\begin{align*}
		\sup_{y \in \cY} \big\{ f_n(y) - \lambda c(x, y) \big\} = \sup_{y \in B_r(x)} \big\{ f_n(y) - \lambda c(x, y) \big\} \leq \sup_{y \in B_r(x)} f_n(y),
	\end{align*}

	Consider  $k > 0$ such that $\hat{\mu} (B^c_k(0)) \leq \delta$. When $n$ is sufficiently large, Dini's theorem shows that $f_n \mathbf{1}_{B_{k+r}(0)} \leq \delta$ by uniform convergence on compact sets. Then 
	\begin{align}
		\sup_{y \in \cY} \big\{ f_n(y) - \lambda c(x, y) \big\}  \leq \left\{\begin{array}{lll}
			& \delta, \text{ if $x \in B_k(0)$}, \\
			& m, \text{ if $x \in B^c_k(0)$}.
		\end{array}\right.
	\end{align}
	Therefore, 
	\begin{align*}
		D(f_n) \leq  \inf_{ \lambda \geq 0 } \lambda \varepsilon + \int_{\cX} \sup_{y \in \cY} \big\{ f_n(y) - \lambda c(x, y) \big\} \hat{\mu}(dx) \leq \delta + \delta \hat{\mu} (B_k(0)) + m \hat{\mu} (B^c_k(0)) \leq \delta + \delta + \delta m.
	\end{align*}
	As $\delta$ is arbitrary, we obtain $D(f_n) \downarrow 0$, which is $D(0)$.
	
	\noindent {\bf Proof of Property 3 and 4.} Since $D(0) = 0$, then by definition of $D^*_{C_b}(\nu)$ in property 3, $D^*_{C_b}(\nu) \geq \int 0 d\nu - D(0) = 0$. Note $C_b$ is a subset of $U_b$, then $D^*_{C_b} (\nu) \leq D^*_{U_b} (\nu)$. To prove $D^*_{U_b} (\nu)$ is no greater than the right-hand side (RHS) of \eqref{eq:Dc}, we only need to consider the case when $\nu$ is a probability measure. Otherwise the RHS is infinity and there is nothing to prove. By the definition of $D^*_{U_b} (\nu)$,
	\begin{align*}
		D^*_{U_b} (\nu) & = \sup_{f \in U_b} \left[ \int_\cY f(y) \nu(dy) - \inf_{ \lambda \geq 0, \; \gamma \in \Gamma} \left( \lambda \varepsilon + \int_{\cX} F(x; \lambda, \gamma) \hat{\mu}(dx) \right) \right] \\
		& = \sup_{\lambda \geq 0, \; \gamma \in \Gamma} \sup_{f \in U_b} \left[ \int_\cY f(y) \nu(dy) - \int_{\cX} F(x; \lambda, \gamma) \hat{\mu}(dx) - \lambda \varepsilon \right].
	\end{align*}
	By the definition of $F$, one has
	\begin{equation*}
		F(x; \lambda, \gamma) \geq f(y) - \lambda c(x, y) + \gamma(x, y), \quad \forall x \in \cX, \, y \in \cY.
	\end{equation*} 
	Hence,
	\begin{equation*}
		\int F(x; \lambda, \gamma) d \pi \geq \int f(y) - \lambda c(x, y) + \gamma(x, y) d\pi, \quad \forall \pi \in \Pi(\hat{\mu}, \nu).
	\end{equation*}
	Rearranging the terms and noting the margin condition, we have
	\begin{equation*}
		\int_\cY f(y) \nu(dy) - \int_{\cX} F(x; \lambda, \gamma) \hat{\mu}(dx) \leq \lambda \int c(x, y) d\pi - \int \gamma(x, y) d\pi, \quad \forall \pi \in \Pi(\hat{\mu}, \nu).
	\end{equation*}
	Therefore, 
	\begin{align*}
		D^*_{U_b} (\nu) \leq \sup_{\lambda \geq 0, \; \gamma \in \Gamma} \left[ \lambda \int c(x, y) d\pi - \int \gamma(x, y) d\pi - \lambda \varepsilon \right], \quad \forall \pi \in \Pi(\hat{\mu}, \nu).
	\end{align*}
	Taking the infimum over $\pi \in \Pi(\hat{\mu}, \nu)$, we obtain $D^*_{U_b} (\nu) \leq \text{RHS}$ of \eqref{eq:Dc}.
	
	To prove $\text{RHS} \leq D^*_{C_b} (\nu)$, first note that if $\nu$ is not a probability measure, then one can take $f = m$ for some constant $m$ and obtain
	\begin{align*}
		D^*_{C_b} (\nu) \geq \sup_{m \in \R} \left[ \int m d\nu - m \right] = \infty.
	\end{align*} 
	
	Furthermore, we can refine the topology of $\cX$ in the same spirit of \citet[Lemma 6.1]{acciaio2021cournot}. By \citet[Theorem 13.11]{kechris2012classical}, there exists a stronger Polish topology $\hat{\cT}^\cX \supseteq \cT^\cX$ with the same Borel sets, such that the mapping $x_{1:t} \mapsto \hat{\mu}(dx_{t+1:T}|x_{1:t})$ is continuous given the $\hat{\cT}^\cX$-topology. Note that $\gamma \in \Gamma$ is $\hat{\cT}^\cX \times \cT^\cY$-continuous after strengthening the topology. Also, $\Gamma$ is a convex subset.
	
	Let $\nu$ be a probability measure. 
	By \citet[Theorem 1.7]{santambrogio2015optimal}, $\Pi(\hat{\mu}, \nu)$, the set of transport plans between $\hat{\mu}$ and $\nu$, is compact under the weak topology. Then by Sion's minimax theorem, we have 
	\begin{align*}
		& \inf_{ \pi \in \Pi(\hat{\mu}, \nu)}   \sup_{ \lambda \geq 0, \; \gamma \in \Gamma} \Big[ \lambda \big( \int c(x, y) d\pi - \varepsilon \big) - \int \gamma(x, y) d\pi \Big] \\
		& =  \sup_{ \lambda \geq 0, \; \gamma \in \Gamma} \inf_{ \pi \in \Pi(\hat{\mu}, \nu)} \Big[ \lambda \big( \int c(x, y) d\pi - \varepsilon \big) - \int \gamma(x, y) d\pi \Big].
	\end{align*}
	Indeed, since $c(x, y)$ is l.s.c. and $\gamma$ is continuous, then the objective is l.s.c. in $\pi$. Besides, the objective is linear and thus quasi-convex in $\pi$. Moreover, the objective is continuous and linear (and thus quasi-concave) in $\lambda$ and $\gamma$.
	
	Since by definition,
	\begin{align*}
		D^*_{C_b} (\nu) & = \sup_{\lambda \geq 0, \; \gamma \in \Gamma} \sup_{f \in C_b} \left[ \int_\cY f(y) \nu(dy) - \int_{\cX} F(x; \lambda, \gamma) \hat{\mu}(dx) - \lambda \varepsilon \right],
	\end{align*}
	we only have to prove
	\begin{equation}\label{eq:pif}
		\inf_{ \pi \in \Pi(\hat{\mu}, \nu)} \Big[ \lambda \int c(x, y) d\pi - \int \gamma(x, y) d\pi \Big] \leq \sup_{f \in C_b} \left[ \int_\cY f(y) \nu(dy) - \int_{\cX} F(x; \lambda, \gamma) \hat{\mu}(dx) \right].
	\end{equation}
	By the classic Kantorovich duality \citep[Chapter 5, Theorem 5.9]{villani2009optimal}, there exists $f(y)$ and $g(x)$ such that $f(y) + g(x) \leq \lambda c(x,y) - \gamma(x, y)$ and  
	\begin{align}\label{eq:lowbd}
		\int_\cY f(y) d \nu + \int_\cX g(x) d \hat{\mu} \geq \left\{\begin{array}{lll}
			& 1/\delta, \text{ if $\inf_{ \pi \in \Pi(\hat{\mu}, \nu)} \int \big[ \lambda c(x, y) - \gamma(x, y)  \big] d\pi = \infty$}, \\
			& 	\inf_{ \pi \in \Pi(\hat{\mu}, \nu)} \int \big[ \lambda c(x, y) - \gamma(x, y) \big] d\pi  - \delta, \text{ otherwise}.
		\end{array}\right.
	\end{align}
	Observing that $F(x; \lambda, \gamma) = \sup_{y \in \cY}  \big\{f(y) - \lambda c(x, y) + \gamma(x, y) \big\} \leq - g(x)$, then
	\begin{equation*}
		\int_\cY f(y) d \nu - \int_\cX F(x; \lambda, \gamma) d \hat{\mu} \geq \int_\cY f(y) d \nu + \int_\cX g(x) d \hat{\mu}. 
	\end{equation*}
	Combining with \eqref{eq:lowbd} and letting $\delta \rightarrow 0$, we prove \eqref{eq:pif} and thus the $\text{RHS} \leq D^*_{C_b} (\nu)$. 
\end{proof}

\begin{proof}[Proof of Corollary \ref{cor:dual}]
		Since $f$ is u.s.c. and $(\cX, \cT^\cX)$ and $(\cY, \cT^\cY)$ are compact, $f$ is bounded from above. There exists a nonincreasing sequence of continuous (and thus bounded on compact $\cX$ and $\cY$) functions $\{ f_n \}^\infty_{n=1}$ such that $f_n \downarrow f$ pointwise. A careful check of the proof for Theorem \ref{thm:cap-dual} indicates that if $\cX$ and $\cY$ are compact, then condition 3 is not needed since the Dini's theorem is no longer needed. Then we apply Theorem \ref{thm:cap-dual} to $f_n$, we obtain
		\begin{equation*}
			J(f_n) := J(\varepsilon; \hat{\mu}) = D(\varepsilon; \hat{\mu})  =: D(f_n).
		\end{equation*}
		Moreover, there exists a primal optimizer $\nu_n \in \cB_{\varepsilon}$ for $J(f_n)$. As $c(x, y)$ is l.s.c. and bounded from below, by \citet[Theorem 4.3]{acciaio2021cournot}, $\cW_c(\hat{\mu}, \nu_n)$ is attained by some $\pi_n \in \Pi_c(\hat{\mu}, \nu_n)$. Since $\cX$ and $\cY$ are compact, $\{ \pi_n \}^\infty_{n=1}$ is tight. By Prokhorov's theorem, there exists a subsequence $\{ \pi_{n_k}\}^\infty_{k=1}$ weakly converging to some $\pi^* \in \cP(\cX \times \cY)$. Denote the marginal of $\pi^*$ on $\cY$ as $\nu^*$. We check that $\nu_{n_k}$ also weakly converges to $\nu^*$ and $\nu^* \in \cB_\varepsilon$.
		
		For any $\alpha \in C_b(\cY)$, we have
		\begin{align*}
			\lim_{k \rightarrow \infty} \int_{\cY} \alpha(y) \nu_{n_k}(dy) = \lim_{k \rightarrow \infty} \int_{\cX \times \cY} \alpha(y) \pi_{n_k}(dx, dy) = \int_{\cX \times \cY} \alpha(y) \pi^*(dx, dy) = \int_{\cY} \alpha(y) \nu^*(dy).
		\end{align*}
		The first equality uses $\pi_{n_k} \in \Pi_c(\hat{\mu}, \nu_{n_k})$. The second equality is due to the weak convergence of $\pi_{n_k} \Rightarrow \pi^*$. The last equality is the definition of the marginal. Then  $\nu_{n_k}$ weakly converges to $\nu^*$. Similarly, we can verify the marginal of $\pi^*$ on $\cX$ is $\hat{\mu}$.
		
		We refine the topology on $\cX$ as in Theorem \ref{thm:cap-dual} such that $\gamma \in \Gamma$ is continuous. By the definition of $\pi_{n}$, we have
		\begin{align*}
			& \int_{\cX \times \cY} c(x, y) \pi_{n}(dx, dy)  \leq \varepsilon, \quad \int_{\cX \times \cY} \gamma(x, y) \pi_{n}(dx, dy) = 0, \quad \forall \gamma \in \Gamma.
		\end{align*}
		There exists a nondecreasing sequence of continuous functions $\{c_n\}^\infty_{n=1}$ such that $c_n \uparrow c$. By monotone convergence theorem and weak convergence,
		\begin{align*}
			& \int_{\cX \times \cY} c(x, y) \pi^*(dx, dy) = \int_{\cX \times \cY} \lim_{n \rightarrow \infty} c_n(x, y) \pi^*(dx, dy) = \lim_{n \rightarrow \infty} \int_{\cX \times \cY}  c_n(x, y) \pi^*(dx, dy) \\
			& =  \lim_{n \rightarrow \infty} \lim_{k \rightarrow \infty} \int_{\cX \times \cY}  c_n(x, y) \pi_{n_k}(dx, dy).
		\end{align*}
		Since the sequence $\{c_n\}$ is nondecreasing,
		\begin{align*}
			& \lim_{k \rightarrow \infty} \int_{\cX \times \cY}  c_n(x, y) \pi_{n_k}(dx, dy) \leq \liminf_{k \rightarrow \infty} \int_{\cX \times \cY}  c_{n_k}(x, y) \pi_{n_k}(dx, dy) \leq \varepsilon.
		\end{align*}
		Hence, $\int_{\cX \times \cY} c(x, y) \pi^*(dx, dy) \leq \varepsilon$. Similarly, as $\gamma \in \Gamma$ is continuous after refining the topology and bounded, weak convergence with the refined topology yields (up to a subsequence of $\{\pi_{n_k}\}_k$, still denoted as $\{\pi_{n_k}\}_k$)
		\begin{align*}
			& \lim_{k \rightarrow \infty} \int_{\cX \times \cY}  \gamma(x, y) \pi_{n_k}(dx, dy) = \int_{\cX \times \cY}  \gamma(x, y) \pi^*(dx, dy) = 0.
		\end{align*} 
		Therefore, $\pi^* \in \Pi_c (\hat{\mu}, \nu^*)$. As a result, we obtain $\cW_c(\hat{\mu}, \nu^*) \leq \varepsilon$ which implies $\nu^* \in \cB_{\varepsilon}$. It also proves that $\cB_\varepsilon$ is compact.
		
		Then 
		\begin{align}
			J(f) & = \sup_{\nu \in \cB_\varepsilon} \int_\cY f(y) \nu(dy)  \geq \int_\cY f(y) \nu^*(dy)  \label{temp1} \\
			& = \int_\cY \lim_{n \rightarrow \infty} f_n(y)  \nu^*(dy) =  \lim_{n \rightarrow \infty} \int_\cY f_n(y)  \nu^*(dy) = \lim_{n \rightarrow \infty} \lim_{k \rightarrow \infty} \int_\cY f_n(y)  \nu_{n_k}(dy). \nonumber
		\end{align}
		Since the sequence $\{ f_n\}^\infty_{n=1}$ is decreasing,
		\begin{align}
			\lim_{k \rightarrow \infty} \int_\cY f_n(y)  \nu_{n_k}(dy) \geq \limsup_{k \rightarrow \infty} \int f_{n_k} (y) \nu_{n_k}(dy) = \limsup_{k \rightarrow \infty} D(f_{n_k}) \geq D(f). \label{temp2}
		\end{align}
		The last inequality uses the fact that $D(f)$ is increasing, as shown in Theorem \ref{thm:cap-dual}. It follows that $J(f) \geq D(f)$. As the weak duality $D(f) \geq J(f)$ holds, we obtain $J(f) = D(f)$ and inequalities in \eqref{temp1}--\eqref{temp2} are equalities. Thus, $\nu^*$ is a primal optimizer.
\end{proof}

\begin{proof}[Proof of Theorem \ref{thm:struct-dual}]
	We deal with the general $f(y)$ and $c(x,y)$ satisfying Assumption \ref{assum:lsc} directly. Fixing $\nu \in \cV$, we introduce the Lagrangian
	\begin{align*}
		\cL (\pi, \lambda, \gamma; \nu) := & \int_{\cX \times \cY} \big( f(y) + \gamma(x, y) \big) \pi(dx, dy) + \lambda \big[ \varepsilon - \int_{\cX \times \cY} c(x, y) \pi(dx, dy) \big].
	\end{align*}
	Note that $\Pi(\hat{\mu}, \nu)$ is compact under the weak topology. We refine the topology on $\cX$ as in Theorem \ref{thm:cap-dual}. $\Pi(\hat{\mu}, \nu)$ is still compact after refining the topology. Since $f$ is u.s.c. and bounded from above, $\gamma(x,y)$ is continuous and bounded, $c(x,y)$ is l.s.c. and non-negative, we obtain that $\cL$ is u.s.c. in $\pi$ by \citet[Lemma 1.6]{santambrogio2015optimal}. Consider the space $C_b(\cX \times \cY)$ endowed with uniform topology and $[0, \infty)$ with the Euclidean topology, one has that $\cL$ is continuous in $\lambda$ and $\gamma$. Similarly, $\cL$ is linear in $\pi$ for given $(\lambda, \gamma)$ and linear in $(\lambda, \gamma)$ for given $\pi$. Sion's minimax theorem shows that
	\begin{equation}\label{eq:struct-minimax}
		\sup_{\pi \in \Pi(\hat{\mu}, \nu)} \inf_{\substack{\lambda \geq 0, \\ \gamma \in \Gamma}} \cL (\pi, \lambda, \gamma; \nu) =  \inf_{\substack{\lambda \geq 0, \\ \gamma \in \Gamma}}  \sup_{\pi \in \Pi(\hat{\mu}, \nu)} \cL (\pi, \lambda, \gamma; \nu).
	\end{equation}
	One can verify that
	\begin{equation*}
		J(\varepsilon; \cV) = \sup_{\nu \in \cV} \sup_{\pi \in \Pi(\hat{\mu}, \nu)} \inf_{\substack{\lambda \geq 0, \\ \gamma \in \Gamma}} \cL (\pi, \lambda, \gamma; \nu).
	\end{equation*}
	Simplifying the RHS of \eqref{eq:struct-minimax}, we obtain $D(\varepsilon; \hat{\mu}, \cV)$ in \eqref{eq:struct-dual} after taking supremum over $\nu \in \cV$. 
\end{proof}

\begin{proof}[Proof of Corollary \ref{cor:equiv}]
	With $c_t(x_t, x_t) = 0$, the assumption $\int_\cX c(x, x) \hat{\mu} (dx) \leq \varepsilon$ in Theorem \ref{thm:cap-dual} or Corollary \ref{cor:dual} holds. We can apply Theorem \ref{thm:cap-dual} or Corollary \ref{cor:dual} to the OT formulation and obtain
	\begin{align*}
		J(\varepsilon; \hat{\mu} ) = D(\varepsilon; \hat{\mu} ) =&  \inf_{ \lambda \geq 0} \quad  \lambda \varepsilon + \int_{\cX} \sup_{y \in \cY}  \big\{f(y) - \lambda c(x, y)\big\} \hat{\mu}(dx) \\
		=&  \inf_{ \lambda \geq 0} \quad  \lambda \varepsilon + \int_{\cX}  \sum^T_{t=1} \sup_{y_t \in \cY_t}  \big\{ f_t(y_t)  - \lambda c_t(x_t, y_t) \big\} \hat{\mu}(dx).
	\end{align*}
	First, the OT primal value is not less than the COT one. For the inverse direction, \citet[Theorem 1]{blanchet2019quantifying} shows that there is a dual optimizer $(\lambda^*, L)$ with a map $L: \cX \rightarrow \cY$, where the assumption $\int_\cX |f(x)| \hat{\mu} (dx) < \infty$ is used. Note that $L$ may not be unique. But it can be chosen measurable, see \citet[Proposition 7.50(b)]{bertsekas1996stochastic}. In view of the separable structure of $f(y)$ and $c(x,y)$, we further know that $L$ is separable with $y_t = L_t(x_t)$. Therefore, $L_t$ does not violate the non-anticipativeness. Then the COT primal value is not less than the OT one. We prove the result as desired.
\end{proof}

\section{Proofs of Section \ref{sec:complex}}

\subsection{Proofs of Section \ref{sec:rad}}
\begin{proof}[Proof of Lemma \ref{lem:Rad}]
	To ease the notation, let $f^c(x, y, \lambda) := f(y) - \lambda c(x, y)$. With the definition of $D(\theta, \varepsilon; \mu \circ \varphi^{-1}_N)$ and $D(\theta, \varepsilon; \hat{\mu} \circ \varphi^{-1}_N) $, by adding and deducting a common term, an application of the triangle inequality gives
	\begin{align*}
		\big|& D(\theta, \varepsilon; \mu \circ \varphi^{-1}_N) - D (\theta, \varepsilon; \hat{\mu} \circ \varphi^{-1}_N) \big|\\
		\leq & \Big| \int \sup_{y \in \cY} \Big\{  \sum^L_{l=1} \sum^{T-1}_{t=1} h_{l,t}(y_{1:t}) \big[ g_{l,t}(\varphi_N(x_{1:T})) - \int g_{l,t}(\varphi_N(x_{1:T})) \mu (d x_{t+1:T} | \varphi^{-1}_N (x_{1:t})) \big] \\
		& \hspace{1.5cm} + f^c(\varphi_N(x), y, \lambda) \Big\} \mu(dx_{1:T}) \\
		& \hspace{0.5cm} - \int \sup_{y \in \cY} \Big\{ \sum^L_{l=1} \sum^{T-1}_{t=1} h_{l, t}(y_{1:t}) \big[ g_{l, t}(\varphi_N(x_{1:T})) - \int g_{l, t}(\varphi_N(x_{1:T})) \mu (d x_{t+1:T} | \varphi^{-1}_N (x_{1:t})) \big] \\
		& \hspace{2.4cm} + f^c(\varphi_N(x), y, \lambda) \Big\} \hat{\mu}(dx_{1:T}| N, \id) \Big| \\
		& + \Big| \int \sup_{y \in \cY} \Big\{ \sum^L_{l=1} \sum^{T-1}_{t=1} h_{l, t}(y_{1:t}) \big[ g_{l,t}(\varphi_N(x_{1:T})) - \int g_{l,t}( \varphi_N(x_{1:T})) \hat{\mu} (d x_{t+1:T} | \varphi^{-1}_N (x_{1:t}), N, \id) \big] \\
		& \hspace{2.0cm} + f^c(\varphi_N(x), y, \lambda) \Big\} \hat{\mu}(dx_{1:T}| N, \id) \\
		&  \hspace{0.5cm} - \int \sup_{y \in \cY} \Big\{ \sum^L_{l=1} \sum^{T-1}_{t=1} h_{l, t}(y_{1:t}) \big[ g_{l, t}(\varphi_N(x_{1:T})) - \int g_{l, t}(\varphi_N(x_{1:T})) \mu (d x_{t+1:T} |\varphi^{-1}_N (x_{1:t})) \big] \\
		&  \hspace{2.4cm} + f^c(\varphi_N(x), y, \lambda) \Big\} \hat{\mu}(dx_{1:T}| N, \id) \Big| \\
		=:& {\rm I} + {\rm II}.
	\end{align*}
	Recalling the definition of $F$, part ${\rm I}$ is the difference between empirical averages and the expected value of $F$. Since the sample is i.i.d. drawn from $\mu$ and continuous functions on compact subsets are bounded, then \citet[Theorem 3.3]{mohri2018foundations} and \citet[Section 7.1]{reppen2020bias} imply
	\begin{align}
		{\rm I} \leq 2 \cR( F( \Theta_k \circ \varphi_N); N) + 2C_{h,g} \sqrt{\frac{\ln(2/\delta)}{2N}} \label{eq:partI}
	\end{align}
	holds with probability at least $1 - \delta$.
	
	For part {\rm II}, we claim that
	\begin{align}
		{\rm II} \leq& C_h \int \sum^L_{l=1} \sum^{T-1}_{t=1} \Big| \int g_{l, t}( \varphi_N(x_{1:T})) \hat{\mu} (d x_{t+1:T} | \varphi^{-1}_N (x_{1:t}), N, \id) \nonumber \\
		& \hspace{2.5cm} - \int g_{l,t}( \varphi_N(x_{1:T})) \mu (d x_{t+1:T} | \varphi^{-1}_N (x_{1:t})) \Big| \hat{\mu}(d x_{1:T} | N, \id) \label{ineq:Lip}\\
		\leq& C_h L_k \sum^{T-1}_{t=1} \sum_{A_t \in \Phi_{N,t}} \hat{\mu} (A_t | N) \Big\{ 2 \cR( \cG_k \circ \varphi_N; n(A_t)) + 2C_g \sqrt{\frac{\ln(2/\delta)}{2n(A_t)}} \Big\} \label{ineq:sub} \\
		=& C_h L_k \sum^{T-1}_{t=1} \sum_{A_t \in \Phi_{N,t}} \hat{\mu} (A_t | N) r(n(A_t); \delta). \nonumber
	\end{align}
	$C_h$ is a large constant to bound all $h_{l, t}$ in $\Theta_k$. $\hat{\mu} (A_t | N)$ is the empirical probability on the set $A_t$ when there are $N$ sample paths. $n(A_t)$ denotes the number of sample paths that fall into the set $A_t$. Thus, we have $n(A_t) = N \hat{\mu} (A_t | N)$ and $\sum_{A_t \in \Phi_{N,t}} n(A_t) = N$. 
	
	Inequality \eqref{ineq:Lip} uses the property that $F$ is Lipschitz in the integral of $g_{l,t}$ with Lipschitz coefficient $C_h$. For inequality \eqref{ineq:sub}, first note that it is the same to use $\hat{\mu} (d x_{t+1:T} | \varphi^{-1}_N (x_{1:t}), N, \varphi_N)$ in the place of $\hat{\mu} (d x_{t+1:T} | \varphi^{-1}_N (x_{1:t}), N, \id)$ in \eqref{ineq:Lip}. By \citet[Lemma 3.3]{backhoff2020estimating}, the law of $\hat{\mu} (d x_{t+1:T} | \varphi^{-1}_N (x_{1:t}), N, \id)$ is the same as that of the empirical measure of $\mu (d x_{t+1:T} | \varphi^{-1}_N (x_{1:t}))$ with sample size $ N \hat{\mu} ( \varphi^{-1}_N (x_{1:t}) | N) $. Applying \citet[Theorem 3.3]{mohri2018foundations} and \citet[Section 7.1]{reppen2020bias} again, we obtain inequality \eqref{ineq:sub} with sample size $n(A_t)$ for each $A_t$. Note that at each time $t$, there are at most $N$ subsets $A_t \in \Phi_{N,t}$ with non-zero probabilities. Thus, we have applied at most $N L_k (T-1)$ times of the Rademacher complexity estimates in the inequality \eqref{ineq:sub}.
	
	With the definition of $r(\cdot; \delta)$ in \eqref{def:r}, note that $x r(x; \delta)$ is concave on $x \in [0, \infty)$ and $r(x; \delta)$ is decreasing on $x \in [0, \infty)$ under assumptions. Then
	\begin{align*}
		& \sum_{A_t \in \Phi_{N,t}} \hat{\mu} (A_t | N) r(n(A_t) ; \delta) = \frac{|\Phi_{N,t}|}{N} \sum_{A_t \in \Phi_{N,t}} \frac{N \hat{\mu} (A_t | N)}{|\Phi_{N,t}|} r\big( N \hat{\mu} (A_t | N) ; \delta \big) \\
		& \leq \frac{|\Phi_{N,t}|}{N} \Big( \sum_{A_t \in \Phi_{N,t}} \frac{N \hat{\mu} (A_t | N)}{|\Phi_{N,t}|} \Big) r\Big( \sum_{A_t \in \Phi_{N,t}} \frac{N \hat{\mu} (A_t | N)}{|\Phi_{N,t}|} ; \delta \Big) = r\Big( \frac{N}{|\Phi_{N,t}|} ; \delta \Big) \leq r\left(N^{1 - qd(T-1)} ; \delta \right).
	\end{align*}
	
	 In total, we have applied at most $NL_k(T-1)+1$ times of the Rademacher complexity estimates in \eqref{eq:partI} and \eqref{ineq:sub}. By replacing $\delta$ with $\delta/(NL_k(T-1)+1)$, we can apply the union bound and state that, with probability at least $1-\delta$, the following inequalities hold simultaneously:
	\begin{align*}
			{\rm I} & \leq 2 \cR( F( \Theta_k \circ \varphi_N); N) + C \sqrt{\frac{\ln[2(NL_k(T-1)+1)/\delta]}{N}}, \\
			{\rm II} & \leq C_h L_k (T-1) r\left(N^{1 - qd(T-1)} ; \delta/(NL_k(T-1)+1) \right) \\
			& \leq C \cR( \cG_k \circ \varphi_N ; N^{1 - qd(T-1)}) + C \sqrt{\frac{\ln[2(NL_k(T-1)+1)/\delta]}{N^{1 - qd(T-1)}}},
	\end{align*}	
	where $C > 0$ is a constant independent of the sample size $N$. Since $N^{1 - qd(T-1)} \leq N$, we obtain the result as desired.
 
\end{proof}

\begin{proof}[Proof of Lemma \ref{lem:Ninf}]
	Since the domain for $\theta$ is compact and the dual objective $D$ is continuous on $\theta$, there exist minimizers denoted by 
	\begin{align*}
		\theta^*_{k, N} \in \arg\min_{\theta \in \Theta_k} D(\theta, \varepsilon; \hat{\mu} \circ \varphi^{-1}_N), \quad  \theta^*_k \in \arg\min_{\theta \in \Theta_k} D(\theta, \varepsilon; \mu).
	\end{align*}
	Since $\sup_{u \in [a,b]^d} |u - \varphi_N(u)| \leq c/N^q$ for some constant $c$, for a given $\eta > 0$, there exists $N_\eta$ such that when $N \geq N_\eta$,
	\begin{align}
		|D(\theta^*_k, \varepsilon; \mu \circ \varphi^{-1}_N) - D(\theta^*_k, \varepsilon; \mu)| &\leq \eta/2, \label{eq:k-phi}\\ 
		|D(\theta^*_{k, N}, \varepsilon; \mu \circ \varphi^{-1}_N) - D(\theta^*_{k, N}, \varepsilon; \mu)| &\leq \eta/2. \label{eq:kN-phi}
	\end{align}
	The definition of $\theta^*_{k,N}$ implies
	\begin{equation}\label{eq:kN-def}
		D(\theta^*_{k, N}, \varepsilon; \hat{\mu} \circ \varphi^{-1}_N) \leq D(\theta^*_k, \varepsilon; \hat{\mu} \circ \varphi^{-1}_N).
	\end{equation}
	Furthermore, for a given $\delta > 0$, Lemma \ref{lem:Rad} proves that
	\begin{align}
		D(\theta^*_k, \varepsilon; \hat{\mu} \circ \varphi^{-1}_N) 
		& \leq D(\theta^*_k, \varepsilon; \mu \circ \varphi^{-1}_N) + 2 \cR( F( \Theta_k \circ \varphi_N); N)  + C \cR( \cG_k \circ \varphi_N ; N^{1 - qd(T-1)}) \nonumber \\
		& \quad + C \sqrt{\frac{\ln[2(NL_k(T-1)+1)/\delta]}{N^{1 - qd(T-1)}}} \nonumber \\
		& =: D(\theta^*_k, \varepsilon; \mu \circ \varphi^{-1}_N) + \cR_N + C \sqrt{\frac{\ln[2(NL_k(T-1)+1)/\delta]}{N^{1 - qd(T-1)}}} \label{ineq:k-rad}
	\end{align}
	holds with probability at least $1 - \delta$. Constant $C$ is independent of sample size $N$. Combining \eqref{eq:kN-def}, \eqref{ineq:k-rad}, and \eqref{eq:k-phi}, one has
	\begin{align}
		D(\theta^*_{k,N}, \varepsilon; \hat{\mu} \circ \varphi^{-1}_N) \leq D(\theta^*_k, \varepsilon; \mu) + \cR_N + C \sqrt{\frac{\ln[2(NL_k(T-1)+1)/\delta]}{N^{1 - qd(T-1)}}} + \eta/2 \label{ineq1}
	\end{align}
	with probability at least $1 - \delta$.
	
	In a similar manner, by invoking the definition of $\theta^*_k$, \eqref{eq:kN-phi}, and Lemma \ref{lem:Rad}, we get
	\begin{align}
		D(\theta^*_k, \varepsilon; \mu) \leq& D(\theta^*_{k,N}, \varepsilon; \mu) \leq D(\theta^*_{k,N}, \varepsilon; \mu \circ \varphi^{-1}_N) + \eta/2 \nonumber \\
		\leq & D(\theta^*_{k,N}, \varepsilon; \hat{\mu} \circ \varphi^{-1}_N) + \cR_N + C \sqrt{\frac{\ln[2(NL_k(T-1)+1)/\delta]}{N^{1 - qd(T-1)}}} + \eta/2  \label{ineq2}
	\end{align}
	with probability at least $1 - \delta$.
	
	We apply the union bound with $\delta$ replaced by $\delta/2$. In view of \eqref{ineq1} and \eqref{ineq2}, with probability at least $1 - \delta$, we have
	\begin{equation}
		|D(\theta^*_{k,N}, \varepsilon; \hat{\mu} \circ \varphi^{-1}_N) - D(\theta^*_k, \varepsilon; \mu)| \leq \cR_N + C \sqrt{\frac{\ln[4(NL_k(T-1)+1)/\delta]}{N^{1 - qd(T-1)}}}  + \eta/2.
	\end{equation}
	For sufficiently large sample size $N$, we then pick $\delta = \delta_N>0$ such that $	C \sqrt{\frac{\ln[4(NL_k(T-1)+1)/\delta]}{N^{1 - qd(T-1)}}} = \eta/2$ or equivalently
	\begin{equation*}
	  \delta_N = 4(NL_k(T-1)+1) \exp \Big[-\frac{\eta^2 N^{1 - qd(T-1)}}{ 4 C^2} \Big].
	\end{equation*}
	Since $\sum_N \delta_N < \infty$, the Borel-Cantelli lemma implies that 
	\begin{equation*}
		\limsup_{N \rightarrow \infty} | D(\theta^*_{k,N}, \varepsilon; \hat{\mu} \circ \varphi^{-1}_N) - D(\theta^*_k, \varepsilon; \mu) | \leq \lim_{N \rightarrow \infty}  \cR_N + \eta/2 + \eta/2 = \eta,
	\end{equation*}
	with probability one. Since $\eta$ is arbitrary, we prove the claim as desired.
\end{proof}

\begin{proof}[Proof of Theorem \ref{thm:conv}]
	For any $\eta > 0$, we can find $\theta_\eta$ such that $D(\theta_\eta, \varepsilon; \mu) \leq \inf_{\theta \in \Theta} D(\theta, \varepsilon; \mu)  + \eta$. By assumption, there exists a sequence of $\theta_k \in \Theta_k$ approximating $\theta_\eta$ pointwise. Observing that $F(x; \lambda, h, g)$ is continuous in $x$ since it is the supremum of a jointly continuous function (noting that $h_{l, t}, g_{l, t}, f, c(x,y)$ are continuous) on a compact domain $\cX \times \cY$. The dominated convergence theorem shows that
	\begin{equation*}
		\lim_{k \rightarrow \infty} D(\theta_k, \varepsilon; \mu) = D(\theta_\eta, \varepsilon; \mu).
	\end{equation*}
	Then
	\begin{equation}\label{eq:inf}
		\limsup_{k \rightarrow \infty} \inf_{\theta \in \Theta_k} D(\theta, \varepsilon; \mu) \leq \lim_{k \rightarrow \infty} D(\theta_k, \varepsilon; \mu) = D(\theta_\eta, \varepsilon; \mu) \leq \inf_{\theta \in \Theta} D(\theta, \varepsilon; \mu)  + \eta.
	\end{equation}
	The opposite direction $\inf_{\theta \in \Theta} D(\theta, \varepsilon; \mu) \leq \inf_{\theta \in \Theta_k} D(\theta, \varepsilon; \mu)$ always holds. Since $\eta$ is arbitrary in \eqref{eq:inf}, we obtain the claim as desired.
\end{proof}

\subsection{Finite sample guarantee}\label{sec:finite}

For the dual problem of SCOT (or COT), we also have a finite sample guarantee as follows. Denote $J^* := \int_\cX f(x) \mu(dx)$ as the expected value with the true $\mu$. Let $\hat{J}_N(\varepsilon) := \sup_{\nu \in \cB_\varepsilon \cap \cV} \int_\cY f(y) \nu(dy) $ with $\cB_{\varepsilon} = \{ \nu: \cW_c (\hat{\mu}(\cdot| N, \varphi_N), \nu) \leq \varepsilon \}$. With high probability, $\hat{J}_N(\varepsilon)$ is an upper bound for $J^*$ under certain conditions.

As a preparation, define a rate function as in \cite{backhoff2020estimating}:
\begin{equation*}
	\text{rate}(N) = \left\{
	\begin{array}{rll}
		& N^{-1/(T+1)},  & d = 1, \\
		& N^{-1/(2T)} \ln(N+1), & d =2, \\
		& N^{-1/(dT)}, & d \geq 3. 
	\end{array}\right.
\end{equation*}

\begin{corollary}
	Suppose
	\begin{enumerate}
		\item $\cV$ is large enough to include the true $\mu$;
		\item the cost $c(x, y) = |x -  y|$;
		\item Assumption \ref{assum:continuous} holds;
		\item $\mu$ is Lipschitz in the sense of \citet[Assumption 1.4]{backhoff2020estimating}.
	\end{enumerate}
	For a given level $\delta_N \in (0, 1)$, let $\varepsilon_N := C \text{rate}(N) + \left( \frac{1}{c N} \ln \left( \frac{2 T}{\delta_N} \right) \right)^{1/2}$ with constants $c$ and $C$ from \citet[Theorem 1.7]{backhoff2020estimating}. Then 
	\begin{equation}
		\p[ J^* \leq \hat{J}_N(\varepsilon_N) ] \geq 1 - \delta_N,
	\end{equation}
	where $\p := \mu^N $ is the product of $N$ true measures. If $\sum^\infty_{N = 1} \delta_N < \infty$ and $\lim_{N \rightarrow \infty} \varepsilon_N = 0$, then any sequence $\hat{\nu}_N$  satisfying $\cW_c (\hat{\mu}(\cdot| N, \varphi_N), \hat{\nu}_N)$ $ \leq \delta_N$, converges under Wasserstein metric $\cW$ to $\mu$ almost surely. Moreover, $ \hat{J}_N(\varepsilon_N) \downarrow J^*$ almost surely.
\end{corollary}
\begin{proof}
	Denote $\mathcal{AW}$ as the adapted Wasserstein distance in \cite{backhoff2020estimating}. Since for any two measures, their $\mathcal{AW} \geq \cW_c \geq \cW$. If $\p[\mathcal{AW} (\hat{\mu}(\cdot| N, \varphi_N), \mu) \geq \varepsilon_N] \leq \delta_N$, then one has $\p[\cW_c (\hat{\mu}(\cdot| N, \varphi_N), \mu) \geq \varepsilon_N] \leq \delta_N$. In other words, with probability at least $1 - \delta_N$, we have $\mu \in \cB_{\varepsilon_N}$. Since $\mu \in \cV$ by assumption, we have $J^* \leq \hat{J}_N(\varepsilon_N)$. With constants $c, C, \varepsilon$ in \citet[Theorem 1.7]{backhoff2020estimating}, we can set
	\begin{equation}
		\delta_N = 2T \exp(-cN\varepsilon^2), \quad C \text{rate}(N) + \varepsilon = \varepsilon_N. 
	\end{equation}
	Canceling $\varepsilon$, we obtain the representation of $\varepsilon_N$. The remaining claim follows similarly as in \citet[Theorem 3.6]{mohajerin2018data}.
\end{proof}

\subsection{Proofs of Section \ref{sec:RegNN}}
	\begin{proof}[Proof of Lemma \ref{lem:Rad_Gk}]
	We first calculate the empirical Rademacher complexity $\cR_S (\cG_k \circ \varphi_N; N)$. By definition, given a sample $S = (x^1_{1:T}, ... , x^N_{1:T} )$ and the function class $\cA_{k, T}$ in \eqref{eq:Ak},
	\begin{align*}
		\cR_S (\cG_k \circ \varphi_N; N) & = \frac{1}{N} \E_{\sigma} \left[ \sup_{g \in \cA_{k, T}} \sum^N_{i=1} \sigma_i \left( \sum^{n_k}_{j=1} w_j \phi(u_j \cdot \varphi_N(x^i_{1:T}) + \vartheta_j) \right) \right] \\
		& = \frac{1}{N} \E_{\sigma} \left[ \sup_{\substack{\sum^{n_k}_{j=1} |w_j| \leq C_{w, k}, \, \sum^{n_k}_{j=1} |\vartheta_j| \leq C_{\vartheta, k}, \\ \|u_j\|_2 \leq C_{u, k}}}  \sum^{n_k}_{j=1} w_j  \sum^N_{i=1} \sigma_i \phi(u_j \cdot \varphi_N(x^i_{1:T}) + \vartheta_j) \right] \\
		& = \frac{C_{w, k}}{N}  \E_{\sigma} \left[ \sup_{\substack{ \sum^{n_k}_{j=1} |\vartheta_j| \leq C_{\vartheta, k}, \|u_j\|_2 \leq C_{u, k}}} \max_{1 \leq j \leq n_k} \left| \sum^N_{i=1} \sigma_i \phi(u_j \cdot \varphi_N(x^i_{1:T}) + \vartheta_j) \right|  \right] \\
		& = \frac{C_{w, k}}{N}  \E_{\sigma} \left[ \sup_{\substack{ |\vartheta| \leq C_{\vartheta, k}, \|u\|_2 \leq C_{u, k}}} \left| \sum^N_{i=1} \sigma_i \phi(u \cdot \varphi_N(x^i_{1:T}) + \vartheta) \right|  \right] \\
		& \leq \frac{C_{w, k} L_\phi }{N}  \E_{\sigma} \left[ \sup_{\substack{ |\vartheta| \leq C_{\vartheta, k}, \|u\|_2 \leq C_{u, k}}} \left| \sum^N_{i=1} \sigma_i \Big(u \cdot \varphi_N(x^i_{1:T}) + \vartheta \Big) \right|  \right] \\
		& \leq \frac{C_{w, k} L_\phi }{N}  \E_{\sigma} \left[ \sup_{|\vartheta| \leq C_{\vartheta, k}} \Big| \sum^N_{i=1} \sigma_i \vartheta \Big|  \right] + \frac{C_{w, k} L_\phi }{N}  \E_{\sigma} \left[ \sup_{\|u\|_2 \leq C_{u, k}} \left| \sum^N_{i=1} \sigma_i u \cdot \varphi_N(x^i_{1:T}) \right|  \right].
	\end{align*}
	The third equality holds by assigning the full weight $w_j$ to the term with the largest absolute value and matching its sign. The first inequality follows from Talagrand's lemma \cite[Lemma 5.7 and Exercise 3.11 (b)]{mohri2018foundations}.
	
	For the first part, 
	\begin{align*}
		\E_{\sigma} \left[ \sup_{|\vartheta| \leq C_{\vartheta, k}} \Big| \sum^N_{i=1} \sigma_i \vartheta \Big| \right] & \leq C_{\vartheta, k} \E_{\sigma} \left[ \Big| \sum^N_{i=1} \sigma_i \Big| \right]  \leq C_{\vartheta, k} \E_{\sigma} \left[ \Big( \sum^N_{i=1} \sigma_i \Big)^2 \right]^{1/2} \\
		& = C_{\vartheta, k} \E_{\sigma} \Big[ \sum^N_{i=1} \sigma^2_i \Big]^{1/2} = C_{\vartheta, k} \sqrt{N},
	\end{align*}
	where the independence of $\sigma_i$ is used.
	
	For the second part,
	\begin{align*}
		& \E_{\sigma} \left[ \sup_{\|u\|_2 \leq C_{u, k}} \left| \sum^N_{i=1} \sigma_i u \cdot \varphi_N(x^i_{1:T}) \right|  \right] \\
		& \quad = C_{u, k} \E_{\sigma} \left[ \left\| \sum^N_{i=1} \sigma_i \varphi_N(x^i_{1:T}) \right\|_2  \right] \leq C_{u, k} \E_{\sigma} \left[ \left\| \sum^N_{i=1} \sigma_i \varphi_N(x^i_{1:T}) \right\|^2_2  \right]^{1/2} \\
		& \quad = C_{u, k} \sqrt{\sum^{N}_{i, j = 1} \E[\sigma_i \sigma_j]   \varphi_N(x^i_{1:T}) \cdot \varphi_N(x^j_{1:T}) } \\
		& \quad = C_{u, k} \sqrt{\sum^{N}_{i = 1} \E[\sigma^2_i]   \|\varphi_N(x^i_{1:T})\|^2_2} = C_{u, k} \sqrt{\sum^{N}_{i = 1}   \|\varphi_N(x^i_{1:T})\|^2_2} \leq C_{u, k} C_{b, a, T, d}  \sqrt{N}.
	\end{align*}
	The first equality is from the equality case in the Cauchy-Schwarz inequality. The last inequality is due to the boundedness of $\varphi_N(x^i_{1:T})$.
	
	In summary, 
	\begin{align}
		\cR_S (\cG_k \circ \varphi_N; N) \leq \frac{C_{w, k} L_\phi C_{\vartheta, k} }{\sqrt{N}} + \frac{C_{w, k} L_\phi C_{u, k} C_{b, a, T, d} }{\sqrt{N}}.
	\end{align}
	Since the right-hand side does not depend on the sample $S$, the inequality also holds for the Rademacher complexity $\cR (\cG_k \circ \varphi_N; N)$.
\end{proof}

\begin{proof}[Proof of Lemma \ref{lem:Rad_F}]
	We still compute the empirical Rademacher complexity first. By definition,
	\begin{align*}
		& \cR_S( F( \Theta_k \circ \varphi_N); N) \\
		& = \frac{1}{N} \E_{\sigma} \left[ \sup_{ \substack{g_{l, t} \in \cA_{k, T}, \; h_{l, t} \in \cA_{k, t}, \lambda \in [0, \lambda_k]} } \sum^N_{i=1} \sigma_i F(\varphi_N(x^i) ; \lambda, \gamma) \right] \\
		& = \frac{1}{N} \E_{\sigma} \Big[ \sup_{ \substack{g_{l, t} \in \cA_{k, T}, \; h_{l, t} \in \cA_{k, t}, \lambda \in [0, \lambda_k]} } \sum^N_{i=1} \sigma_i \sup_{y \in \cY} \Big\{ f(y) - \lambda c(\varphi_N(x^i), y)  + \gamma(\varphi_N(x^i), y) \Big\} \Big], 
	\end{align*} 	
	where
	\begin{align*}
		\gamma(\varphi_N(x^i), y) = \sum^{L_k}_{l=1} \sum^{T-1}_{t=1} h_{l,t}(y_{1:t}) \Big[ g_{l,t}(\varphi_N(x^i_{1:T})) - \int g_{l,t}(\varphi_N(x^i_{1:t}), \varphi_N(x_{t+1:T})) \mu (d x_{t+1:T} | \varphi^{-1}_N (x^i_{1:t})) \Big].
	\end{align*}
	
	{\bf Step 1}: We claim that there exists a sample $(y^1_{1:T}, \ldots, y^N_{1:T})$, such that
	\begin{align}
		& \cR_S( F( \Theta_k \circ \varphi_N); N) \nonumber \\
		& \leq \frac{C_c}{N} \E_{\sigma} \Big[ \sup_{\lambda \in [0, \lambda_k]} \sum^N_{i=1} \sigma_i \lambda \Big] \label{ineq:Step1}\\
		& \quad + \sum^{L_k}_{l=1} \sum^{T-1}_{t=1} \frac{1}{N} \E_{\sigma} \Big[ \sup_{g_{l, t} \in \cA_{k, T}, \; h_{l, t} \in \cA_{k, t}} \sum^N_{i=1} \sigma_i h_{l,t}(y^i_{1:t}) g_{l,t}(\varphi_N(x^i_{1:T})) \Big] \nonumber \\
		& \quad + \sum^{L_k}_{l=1} \sum^{T-1}_{t=1} \frac{1}{N} \E_{\sigma} \Big[ \sup_{\substack{g_{l, t} \in \cA_{k, T}, \\ h_{l, t} \in \cA_{k, t}}} \sum^N_{i=1} \sigma_i h_{l,t}(y^i_{1:t}) \int g_{l,t}(\varphi_N(x^i_{1:t}), \varphi_N(x_{t+1:T})) \mu (d x_{t+1:T} | \varphi^{-1}_N (x^i_{1:t})) \Big], \nonumber
	\end{align}
	where $C_c$ is a universal bound for $c(x, y)$, thanks to the continuity and the compactness of the domain.

	By definition,
	\begin{align*}
		& \frac{1}{N} \E_{\sigma} \Big[ \sup_{ \substack{g_{l, t} \in \cA_{k, T}, \; h_{l, t} \in \cA_{k, t}, \lambda \in [0, \lambda_k]} } \sum^N_{i=1} \sigma_i \sup_{y \in \cY} \Big( f(y) - \lambda c(\varphi_N(x^i), y)  + \gamma(\varphi_N(x^i), y) \Big) \Big] \\
		& = \frac{1}{N} \E_{\sigma_{1:N-1}} \Big[ \E_{\sigma_N} \Big[ \sup_{ \substack{g_{l, t} \in \cA_{k, T}, \; h_{l, t} \in \cA_{k, t}, \\ \lambda \in [0, \lambda_k]} } \Big\{ U_{N-1}(g, h, \lambda) \\
		& \hspace{6.5cm} + \sigma_N \sup_{y \in \cY} \Big( f(y) - \lambda c(\varphi_N(x^N), y)  + \gamma(\varphi_N(x^N), y) \Big) \Big\} \Big] \Big],
	\end{align*}
	where 
	\begin{align*}
		U_{N-1}(g, h, \lambda) := \sum^{N-1}_{i=1} \sigma_i \sup_{y \in \cY} \Big( f(y) - \lambda c(\varphi_N(x^i), y)  + \gamma(\varphi_N(x^i), y) \Big). 
	\end{align*}
	
	By definition of the supremum, for any $\varepsilon > 0$, there exist $(g^1, h^1, \lambda^1)$ and $(g^2, h^2, \lambda^2)$ such that
	\begin{align*}
		U_{N-1}(g^1, h^1, \lambda^1) + F(\varphi_N(x^N); \lambda^1, \gamma^1) & \geq \sup_{ \substack{g_{l, t} \in \cA_{k, T}, \\ h_{l, t} \in \cA_{k, t}, \\ \lambda \in [0, \lambda_k]} } \Big\{ U_{N-1}(g, h, \lambda) + F(\varphi_N(x^N); \lambda, \gamma) \Big\} -  \varepsilon, \\
		U_{N-1}(g^2, h^2, \lambda^2) - F(\varphi_N(x^N); \lambda^2, \gamma^2) & \geq \sup_{ \substack{g_{l, t} \in \cA_{k, T}, \\ h_{l, t} \in \cA_{k, t}, \\ \lambda \in [0, \lambda_k]} } \Big\{ U_{N-1}(g, h, \lambda) - F(\varphi_N(x^N); \lambda, \gamma) \Big\} -  \varepsilon,
	\end{align*}
	where $\gamma^1$ and $\gamma^2$ are defined by $(g^1, h^1)$ and $(g^2, h^2)$, respectively.
	
	Hence, for any $\varepsilon > 0$, we have
	\begin{align*}
		& \E_{\sigma_N} \left[ \sup_{ \substack{g_{l, t} \in \cA_{k, T}, h_{l, t} \in \cA_{k, t}, \lambda \in [0, \lambda_k]} } \Big\{ U_{N-1}(g, h, \lambda) + F(\varphi_N(x^N); \lambda, \gamma) \Big\} \right] - \varepsilon \\
		& =   \frac{1}{2}  \sup_{ \substack{g_{l, t} \in \cA_{k, T}, h_{l, t} \in \cA_{k, t}, \lambda \in [0, \lambda_k]} } \Big\{ U_{N-1}(g, h, \lambda) + F(\varphi_N(x^N); \lambda, \gamma) \Big\} - \frac{\varepsilon}{2} \\
		& \quad +  \frac{1}{2} \sup_{ \substack{g_{l, t} \in \cA_{k, T}, h_{l, t} \in \cA_{k, t}, \lambda \in [0, \lambda_k]} } \Big\{ U_{N-1}(g, h, \lambda) - F(\varphi_N(x^N); \lambda, \gamma) \Big\} - \frac{\varepsilon}{2} \\
		&\leq \frac{1}{2} \left( U_{N-1}(g^1, h^1, \lambda^1) + F(\varphi_N(x^N); \lambda^1, \gamma^1) \right) + \frac{1}{2} \left( U_{N-1}(g^2, h^2, \lambda^2) - F(\varphi_N(x^N); \lambda^2, \gamma^2) \right). 
	\end{align*}
	Since the functions in $F$ are continuous in $y$ and the domain $\cY$ is compact, the supremum over $y$ is attained at some $y^\varepsilon$. We emphasize that $y^\varepsilon$ relies on $\varepsilon$, since $(g^1, h^1, \lambda^1)$ is different when $\varepsilon$ changes. Then
	\begin{align*}
		& F(\varphi_N(x^N); \lambda^1, \gamma^1) - F(\varphi_N(x^N); \lambda^2, \gamma^2) \\
		& = \sup_{y \in \cY} \Big( f(y) - \lambda^1 c(\varphi_N(x^N), y)  + \gamma^1(\varphi_N(x^N), y) \Big) - \sup_{y \in \cY} \Big( f(y) - \lambda^2 c(\varphi_N(x^N), y)  + \gamma^2(\varphi_N(x^N), y) \Big) \\
		& \leq \Big( f(y^{\varepsilon}) - \lambda^1 c(\varphi_N(x^N), y^\varepsilon)  + \gamma^1(\varphi_N(x^N), y^\varepsilon) \Big) - \Big( f(y^\varepsilon) - \lambda^2 c(\varphi_N(x^N), y^\varepsilon)  + \gamma^2(\varphi_N(x^N), y^\varepsilon) \Big) \\
		& \leq C_c s_{\lambda, \varepsilon} (\lambda^1 - \lambda^2) + s_{\gamma, \varepsilon} \big[ \gamma^1(\varphi_N(x^N), y^\varepsilon) - \gamma^2(\varphi_N(x^N), y^\varepsilon) \big],
	\end{align*}
	where $|c(\varphi_N(x^N), y^\varepsilon)| \leq C_c$, $s_{\lambda, \varepsilon} = \text{sign} (\lambda^1 - \lambda^2)$, and $s_{\gamma, \varepsilon} = \text{sign}[\gamma^1(\varphi_N(x^N), y^\varepsilon) - \gamma^2(\varphi_N(x^N), y^\varepsilon)]$. The signs $s_{\lambda, \varepsilon}$ and $s_{\gamma, \varepsilon}$ may vary when $\varepsilon$ changes.
	
	The previous two inequalities imply
	\begin{align*}
		& \E_{\sigma_N} \left[ \sup_{ \substack{g_{l, t} \in \cA_{k, T}, h_{l, t} \in \cA_{k, t}, \lambda \in [0, \lambda_k]} } \Big\{ U_{N-1}(g, h, \lambda) + F(\varphi_N(x^N); \lambda, \gamma) \Big\} \right] - \varepsilon \\
		& \leq \frac{1}{2} U_{N-1}(g^1, h^1, \lambda^1) + \frac{1}{2} U_{N-1}(g^2, h^2, \lambda^2) \\
		& \quad + \frac{1}{2} C_c s_{\lambda, \varepsilon} (\lambda^1 - \lambda^2) + \frac{1}{2} s_{\gamma, \varepsilon} \big[ \gamma^1(\varphi_N(x^N), y^\varepsilon) - \gamma^2(\varphi_N(x^N), y^\varepsilon) \big] \\
		& = \frac{1}{2} \left( U_{N-1}(g^1, h^1, \lambda^1) + C_c s_{\lambda, \varepsilon} \lambda^1 +  s_{\gamma, \varepsilon} \gamma^1(\varphi_N(x^N), y^\varepsilon) \right)  \\
		& \quad + \frac{1}{2} \left( U_{N-1}(g^2, h^2, \lambda^2) -  C_c s_{\lambda, \varepsilon}  \lambda^2   -  s_{\gamma, \varepsilon}  \gamma^2(\varphi_N(x^N), y^\varepsilon) \right) \\
		& \leq \frac{1}{2}  \sup_{ \substack{g_{l, t} \in \cA_{k, T}, h_{l, t} \in \cA_{k, t}, \lambda \in [0, \lambda_k]} } \left( U_{N-1}(g, h, \lambda) + C_c s_{\lambda, \varepsilon} \lambda +  s_{\gamma, \varepsilon} \gamma(\varphi_N(x^N), y^\varepsilon) \right)  \\
		& \quad + \frac{1}{2}  \sup_{ \substack{g_{l, t} \in \cA_{k, T}, h_{l, t} \in \cA_{k, t}, \lambda \in [0, \lambda_k]} } \left( U_{N-1}(g, h, \lambda) -  C_c s_{\lambda, \varepsilon}  \lambda   -  s_{\gamma, \varepsilon}  \gamma(\varphi_N(x^N), y^\varepsilon) \right).
	\end{align*}
	We can choose a subsequence, still indexed by $\varepsilon$, such that $s_{\lambda, \varepsilon}$ and $s_{\gamma, \varepsilon}$ do not change when $\varepsilon \rightarrow 0$. Denote them as $s_{\lambda, \varepsilon} = s_\lambda$ and $s_{\gamma, \varepsilon} = s_\gamma$. Dropping to a subsequence of the subsequence, we can assume $y^\varepsilon \rightarrow y^N \in \cY$ when $\varepsilon \rightarrow 0$, thanks to the compactness of $\cY$. Since the right-hand side is a continuous function of $y$, we obtain the following result when $\varepsilon \rightarrow 0$:
	\begin{align*}
		& \E_{\sigma_N} \left[ \sup_{ \substack{g_{l, t} \in \cA_{k, T}, h_{l, t} \in \cA_{k, t}, \lambda \in [0, \lambda_k]} } \Big\{ U_{N-1}(g, h, \lambda) + F(\varphi_N(x^N); \lambda, \gamma) \Big\} \right] \\
		& \leq \frac{1}{2}  \sup_{ \substack{g_{l, t} \in \cA_{k, T}, h_{l, t} \in \cA_{k, t}, \lambda \in [0, \lambda_k]} } \left( U_{N-1}(g, h, \lambda) + C_c s_{\lambda} \lambda +  s_{\gamma} \gamma(\varphi_N(x^N), y^N) \right)  \\
		& \quad + \frac{1}{2}  \sup_{ \substack{g_{l, t} \in \cA_{k, T}, h_{l, t} \in \cA_{k, t}, \lambda \in [0, \lambda_k]} } \left( U_{N-1}(g, h, \lambda) -  C_c s_{\lambda}  \lambda   -  s_{\gamma}  \gamma(\varphi_N(x^N), y^N) \right) \\
		& = \E_{\zeta_N, \xi_N} \left[  \sup_{ \substack{g_{l, t} \in \cA_{k, T}, h_{l, t} \in \cA_{k, t}, \lambda \in [0, \lambda_k]} } \left( U_{N-1}(g, h, \lambda) +  C_c \zeta_N  \lambda  +  \xi_N  \gamma(\varphi_N(x^N), y^N) \right) \right].
	\end{align*}
	In the last equality, $\zeta_N$ and $\xi_N$ are two uniform random variables taking values in $\{ -1, + 1\}$, but not independent of each other.
	
	By repeating the same procedure for $\sigma_{1:N-1}$, we have
	\begin{align*}
		& \cR_S( F( \Theta_k \circ \varphi_N); N) \\
		& \leq \frac{1}{N} \E_{\zeta_{1:N}, \xi_{1:N}} \left[  \sup_{ \substack{g_{l, t} \in \cA_{k, T}, h_{l, t} \in \cA_{k, t}, \\ \lambda \in [0, \lambda_k]} } \left(  C_c \sum^N_{i=1} \zeta_i  \lambda  +  \sum^N_{i=1} \xi_i  \gamma(\varphi_N(x^i), y^i) \right) \right] \\
		& \leq \frac{C_c}{N} \E_{\zeta_{1:N}, \xi_{1:N}} \left[  \sup_{ \substack{\lambda \in [0, \lambda_k]} } \sum^N_{i=1} \zeta_i  \lambda \right] + \frac{1}{N} \E_{\zeta_{1:N}, \xi_{1:N}} \left[  \sup_{ \substack{g_{l, t} \in \cA_{k, T}, \\ h_{l, t} \in \cA_{k, t}} } \left( \sum^N_{i=1} \xi_i  \gamma(\varphi_N(x^i), y^i) \right) \right] \\
		& = \frac{C_c}{N} \E_{\zeta_{1:N}} \left[  \sup_{ \substack{\lambda \in [0, \lambda_k]} } \sum^N_{i=1} \zeta_i  \lambda \right] + \frac{1}{N} \E_{\xi_{1:N}} \left[  \sup_{ \substack{g_{l, t} \in \cA_{k, T}, \\ h_{l, t} \in \cA_{k, t}} } \sum^N_{i=1} \xi_i  \gamma(\varphi_N(x^i), y^i)  \right].
	\end{align*}
	$\zeta_{1:N}$ are independent of each other. Similarly, $\xi_{1:N}$ are also independent of each other. Therefore, the last two terms are empirical Rademacher complexities of the corresponding function class.
	
	By the definition of $\gamma$ and the linear property of empirical Rademacher complexity \cite[Exercise 3.8]{mohri2018foundations}, we obtain
	\begin{align*}
		& \frac{1}{N} \E_{\xi_{1:N}} \left[  \sup_{ \substack{g_{l, t} \in \cA_{k, T}, \\ h_{l, t} \in \cA_{k, t}} } \sum^N_{i=1} \xi_i  \gamma(\varphi_N(x^i), y^i)  \right] \\
		& =  \sum^{L_k}_{l=1} \sum^{T-1}_{t=1} \frac{1}{N} \E_{\xi_{1:N}} \Big[  \sup_{ \substack{g_{l, t} \in \cA_{k, T}, \\ h_{l, t} \in \cA_{k, t}} } \sum^N_{i=1} \xi_i  h_{l,t}(y^i_{1:t}) \Big[ g_{l,t}(\varphi_N(x^i_{1:T})) \\
		& \hspace{7cm} - \int g_{l,t}(\varphi_N(x^i_{1:t}), \varphi_N(x_{t+1:T})) \mu (d x_{t+1:T} | \varphi^{-1}_N (x^i_{1:t})) \Big] \Big] \\
		& \leq \sum^{L_k}_{l=1} \sum^{T-1}_{t=1} \frac{1}{N} \E_{\xi_{1:N}} \Big[  \sup_{ \substack{g_{l, t} \in \cA_{k, T}, h_{l, t} \in \cA_{k, t}} } \sum^N_{i=1} \xi_i  h_{l,t}(y^i_{1:t}) g_{l,t}(\varphi_N(x^i_{1:T})) \Big] \\
		& \quad + \sum^{L_k}_{l=1} \sum^{T-1}_{t=1} \frac{1}{N} \E_{\xi_{1:N}} \Big[  \sup_{ \substack{g_{l, t} \in \cA_{k, T}, h_{l, t} \in \cA_{k, t}} } \sum^N_{i=1} \xi_i h_{l,t}(y^i_{1:t}) \int g_{l,t}(\varphi_N(x^i_{1:t}), \varphi_N(x_{t+1:T})) \mu (d x_{t+1:T} | \varphi^{-1}_N (x^i_{1:t})) \Big]. 
	\end{align*}
	The last inequality used the fact that $- g_{l, t} \in \cA_{k, T}$ whenever $g_{l, t} \in \cA_{k, T}$. We have proved \eqref{ineq:Step1}.

	{\bf Step 2}: Next, we want to show
	\begin{align}
		& \frac{1}{N} \E_{\sigma} \Big[ \sup_{g_{l, t} \in \cA_{k, T}, \; h_{l, t} \in \cA_{k, t}} \sum^N_{i=1} \sigma_i h_{l,t}(y^i_{1:t}) g_{l,t}(\varphi_N(x^i_{1:T})) \Big] \nonumber \\
		& \leq \frac{C_g}{N} \E_{\sigma} \Big[ \sup_{h_{l, t} \in \cA_{k, t}} \sum^N_{i=1} \sigma_i h_{l,t}(y^i_{1:t}) \Big] + \frac{C_h}{N} \E_{\sigma} \Big[ \sup_{g_{l, t} \in \cA_{k, T}} \sum^N_{i=1} \sigma_i g_{l,t}(\varphi_N(x^i_{1:T})) \Big] \label{ineq:h_g}
	\end{align}
	and
	\begin{align}
		& \frac{1}{N} \E_{\sigma} \Big[ \sup_{g_{l, t} \in \cA_{k, T}, \; h_{l, t} \in \cA_{k, t}} \sum^N_{i=1} \sigma_i h_{l,t}(y^i_{1:t}) \int g_{l,t}(\varphi_N(x^i_{1:t}), \varphi_N(x_{t+1:T})) \mu (d x_{t+1:T} | \varphi^{-1}_N (x^i_{1:t})) \Big] \nonumber \\
		& \leq \frac{C_g}{N} \E_{\sigma} \Big[ \sup_{h_{l, t} \in \cA_{k, t}} \sum^N_{i=1} \sigma_i h_{l,t}(y^i_{1:t}) \Big] \label{ineq:h_integ}\\
		& \quad + \frac{C_h}{N} \E_{\sigma} \Big[ \sup_{g_{l, t} \in \cA_{k, T}} \sum^N_{i=1} \sigma_i \int g_{l,t}(\varphi_N(x^i_{1:t}), \varphi_N(x_{t+1:T})) \mu (d x_{t+1:T} | \varphi^{-1}_N (x^i_{1:t})) \Big], \nonumber
	\end{align}
	where $C_g$ and $C_h$ are universal constants to bound $g_{l, t}$ and $h_{l, t}$, respectively.
	
	We only need to prove \eqref{ineq:h_g} since \eqref{ineq:h_integ} follows in the same way. The idea is similar to the Step 1. We give the detail for the completeness of the proof. By definition,
	\begin{align*}
		& \frac{1}{N} \E_{\sigma} \Big[ \sup_{g_{l, t} \in \cA_{k, T}, \; h_{l, t} \in \cA_{k, t}} \sum^N_{i=1} \sigma_i h_{l,t}(y^i_{1:t}) g_{l,t}(\varphi_N(x^i_{1:T})) \Big] \\
		& =  \frac{1}{N} \E_{\sigma_{1:N-1}} \Big[  \E_{\sigma_N} \Big[ \sup_{g_{l, t} \in \cA_{k, T}, \; h_{l, t} \in \cA_{k, t}} \Big( V_{N-1}(g_{l, t}, h_{l, t}) + \sigma_N h_{l,t}(y^N_{1:t}) g_{l,t}(\varphi_N(x^N_{1:T})) \Big) \Big] \Big],
	\end{align*}
	where
	\begin{align*}
		V_{N-1}(g_{l, t}, h_{l, t}) = \sum^{N-1}_{i=1} \sigma_i h_{l,t}(y^i_{1:t}) g_{l,t}(\varphi_N(x^i_{1:T})).
	\end{align*}
	
	For any $\varepsilon > 0$, there exist $(g^1_{l, t}, h^1_{l, t})$ and $(g^2_{l, t}, h^2_{l, t})$ such that
	\begin{align*}
		& V_{N-1}(g^1_{l, t}, h^1_{l, t}) + h^1_{l,t}(y^N_{1:t}) g^1_{l,t}(\varphi_N(x^N_{1:T})) \\
		& \quad \geq \sup_{g_{l, t} \in \cA_{k, T}, \; h_{l, t} \in \cA_{k, t}} \Big( V_{N-1}(g_{l, t}, h_{l, t}) + h_{l,t}(y^N_{1:t}) g_{l,t}(\varphi_N(x^N_{1:T})) \Big) - \varepsilon, \\
		& V_{N-1}(g^2_{l, t}, h^2_{l, t}) -  h^2_{l,t}(y^N_{1:t}) g^2_{l,t}(\varphi_N(x^N_{1:T})) \\
		& \quad \geq \sup_{g_{l, t} \in \cA_{k, T}, \; h_{l, t} \in \cA_{k, t}} \Big( V_{N-1}(g_{l, t}, h_{l, t}) - h_{l,t}(y^N_{1:t}) g_{l,t}(\varphi_N(x^N_{1:T})) \Big) - \varepsilon.
	\end{align*}
	It yields
	\begin{align*}
		& \E_{\sigma_N} \left[ \sup_{g_{l, t} \in \cA_{k, T}, \; h_{l, t} \in \cA_{k, t}} \Big( V_{N-1}(g_{l, t}, h_{l, t}) + \sigma_N h_{l,t}(y^N_{1:t}) g_{l,t}(\varphi_N(x^N_{1:T})) \Big) \right] - \varepsilon \\
		&\leq \frac{1}{2} \left( V_{N-1}(g^1_{l, t}, h^1_{l, t}) + h^1_{l,t}(y^N_{1:t}) g^1_{l,t}(\varphi_N(x^N_{1:T})) \right) + \frac{1}{2} \left( V_{N-1}(g^2_{l, t}, h^2_{l, t}) - h^2_{l,t}(y^N_{1:t}) g^2_{l,t}(\varphi_N(x^N_{1:T})) \right). 
	\end{align*}
	Moreover,
	\begin{align*}
		& h^1_{l,t}(y^N_{1:t}) g^1_{l,t}(\varphi_N(x^N_{1:T})) - h^2_{l,t}(y^N_{1:t}) g^2_{l,t}(\varphi_N(x^N_{1:T})) \\
		& \leq C_h s_{g, \varepsilon} \Big[ g^1_{l,t}(\varphi_N(x^N_{1:T})) - g^2_{l,t}(\varphi_N(x^N_{1:T})) \Big] + C_g s_{h, \varepsilon} \Big[ h^1_{l,t}(y^N_{1:t}) - h^2_{l,t}(y^N_{1:t}) \Big],
	\end{align*}
	where $s_{g, \varepsilon} = \text{sign} [ g^1_{l,t}(\varphi_N(x^N_{1:T})) - g^2_{l,t}(\varphi_N(x^N_{1:T})) ]$ and $s_{h, \varepsilon} = \text{sign}[h^1_{l,t}(y^N_{1:t}) - h^2_{l,t}(y^N_{1:t})]$, respectively.

	The previous two inequalities imply
	\begin{align*}
		& \E_{\sigma_N} \left[ \sup_{g_{l, t} \in \cA_{k, T}, \; h_{l, t} \in \cA_{k, t}} \Big( V_{N-1}(g_{l, t}, h_{l, t}) + \sigma_N h_{l,t}(y^N_{1:t}) g_{l,t}(\varphi_N(x^N_{1:T})) \Big) \right] - \varepsilon \\
		&\leq \frac{1}{2} \left( V_{N-1}(g^1_{l, t}, h^1_{l, t}) + C_h s_{g, \varepsilon} g^1_{l,t}(\varphi_N(x^N_{1:T})) + C_g s_{h, \varepsilon} h^1_{l,t}(y^N_{1:t}) \right) \\
		& \quad + \frac{1}{2} \left( V_{N-1}(g^2_{l, t}, h^2_{l, t}) - C_h s_{g, \varepsilon} g^2_{l,t}(\varphi_N(x^N_{1:T})) - C_g s_{h, \varepsilon} h^2_{l,t}(y^N_{1:t}) \right) \\
		&\leq \frac{1}{2} \sup_{g_{l, t} \in \cA_{k, T}, \; h_{l, t} \in \cA_{k, t}} \left( V_{N-1}(g_{l, t}, h_{l, t}) + C_h s_{g, \varepsilon} g_{l,t}(\varphi_N(x^N_{1:T})) + C_g s_{h, \varepsilon} h_{l,t}(y^N_{1:t}) \right) \\
		& \quad + \frac{1}{2} \sup_{g_{l, t} \in \cA_{k, T}, \; h_{l, t} \in \cA_{k, t}} \left( V_{N-1}(g_{l, t}, h_{l, t}) - C_h s_{g, \varepsilon} g_{l,t}(\varphi_N(x^N_{1:T})) - C_g s_{h, \varepsilon} h_{l,t}(y^N_{1:t}) \right).
	\end{align*}
	Letting $\varepsilon \rightarrow 0$ (along a subsequence) and repeating the same procedure for $\sigma_{1:N-1}$, we obtain  
	\begin{align*}
		& \frac{1}{N} \E_{\sigma} \Big[ \sup_{g_{l, t} \in \cA_{k, T}, \; h_{l, t} \in \cA_{k, t}} \sum^N_{i=1} \sigma_i h_{l,t}(y^i_{1:t}) g_{l,t}(\varphi_N(x^i_{1:T})) \Big] \\
		& \leq \frac{1}{N} \E_{\zeta_{1:N}, \xi_{1:N}} \Big[  \sup_{g_{l, t} \in \cA_{k, T}, \; h_{l, t} \in \cA_{k, t}} \sum^N_{i=1} \left( C_h \zeta_i g_{l,t}(\varphi_N(x^i_{1:T})) + C_g \xi_i h_{l,t}(y^i_{1:t}) \right) \Big] \\
		& \leq \frac{C_h}{N} \E_{\zeta_{1:N}} \Big[ \sup_{g_{l, t} \in \cA_{k, T}} \sum^N_{i=1} \zeta_i g_{l,t}(\varphi_N(x^i_{1:T})) \Big]  + \frac{C_g}{N} \E_{\xi_{1:N}} \Big[ \sup_{h_{l, t} \in \cA_{k, t}} \sum^N_{i=1} \xi_i h_{l,t}(y^i_{1:t}) \Big],
	\end{align*}
	as desired.

	{\bf Step 3}: With the same proof as in Lemma \ref{lem:Rad_Gk}, we can show
	\begin{align}
		\frac{C_c}{N} \E_{\sigma} \Big[ \sup_{\lambda \in [0, \lambda_k]} \sum^N_{i=1} \sigma_i \lambda \Big] & \leq \frac{C_c \lambda_k}{\sqrt{N}}, \label{ineq:Rad1} \\
		\frac{C_h}{N} \E_{\zeta_{1:N}} \Big[ \sup_{g_{l, t} \in \cA_{k, T}} \sum^N_{i=1} \zeta_i g_{l,t}(\varphi_N(x^i_{1:T})) \Big]  & \leq \frac{C_h C_{w, k} L_\phi (C_{\vartheta, k} + C_{u, k} C_{b, a, T, d}) }{\sqrt{N}}, \label{ineq:Rad2} \\
		\frac{C_g}{N} \E_{\xi_{1:N}} \Big[ \sup_{h_{l, t} \in \cA_{k, t}} \sum^N_{i=1} \xi_i h_{l,t}(y^i_{1:t}) \Big] & \leq \frac{C_g C_{w, k} L_\phi (C_{\vartheta, k} + C_{u, k} C_{b, a, T, d}) }{\sqrt{N}}. \label{ineq:Rad3}
	\end{align}
	
	We claim that the last term in \eqref{ineq:h_integ} satisfies 
	\begin{align}
		& \frac{1}{N} \E_{\sigma} \Big[ \sup_{g_{l, t} \in \cA_{k, T}} \sum^N_{i=1} \sigma_i \int g_{l,t}(\varphi_N(x^i_{1:t}), \varphi_N(x_{t+1:T})) \mu (d x_{t+1:T} | \varphi^{-1}_N (x^i_{1:t})) \Big] \label{ineq:diff} \\
		& \leq \frac{C_{b,a,T,d} (L_\cA  L_\mu + L_\cA)}{N^q} + \frac{1}{N} \E_{\sigma} \Big[ \sup_{g_{l, t} \in \cA_{k, T}} \sum^N_{i=1} \sigma_i \int g_{l,t}(\varphi_N(x^i_{1:t}), \varphi_N(x_{t+1:T})) \mu (d x_{t+1:T} | x^i_{1:t}) \Big]. \nonumber
	\end{align}
	Indeed,
	\begin{align}
		& \Big| \E_{\sigma} \Big[ \sup_{g_{l, t} \in \cA_{k, T}} \sum^N_{i=1} \sigma_i \int g_{l,t}(\varphi_N(x^i_{1:t}), \varphi_N(x_{t+1:T})) \mu (d x_{t+1:T} | \varphi^{-1}_N (x^i_{1:t})) \Big] \nonumber \\
		& \quad - \E_{\sigma} \Big[ \sup_{g_{l, t} \in \cA_{k, T}} \sum^N_{i=1} \sigma_i \int g_{l,t}(\varphi_N(x^i_{1:t}), \varphi_N(x_{t+1:T})) \mu (d x_{t+1:T} | x^i_{1:t}) \Big] \Big| \nonumber \\
		& \leq \Big| \E_{\sigma} \Big[ \sup_{g_{l, t} \in \cA_{k, T}} \sum^N_{i=1} \sigma_i \Big( \int g_{l,t}(\varphi_N(x^i_{1:t}), \varphi_N(x_{t+1:T})) \mu (d x_{t+1:T} | \varphi^{-1}_N (x^i_{1:t})) \nonumber \\
		& \hspace{4cm} - \int g_{l,t}(\varphi_N(x^i_{1:t}), \varphi_N(x_{t+1:T})) \mu (d x_{t+1:T} | x^i_{1:t}) \Big) \Big] \Big| \nonumber \\
		& \leq \Big| \E_\sigma \Big[ \sup_{g_{l, t} \in \cA_{k, T}} \Big( \sum^N_{i=1} \sigma^2_i \Big)^{1/2} \Big\{ \sum^N_{i=1} \Big( \int g_{l,t}(\varphi_N(x^i_{1:t}), \varphi_N(x_{t+1:T})) \mu (d x_{t+1:T} | \varphi^{-1}_N (x^i_{1:t})) \nonumber \\
		& \hspace{5.5cm} - \int g_{l,t}(\varphi_N(x^i_{1:t}), \varphi_N(x_{t+1:T})) \mu (d x_{t+1:T} | x^i_{1:t}) \Big)^2 \Big\}^{1/2} \Big] \Big| \nonumber \\
		& \leq \sqrt{N} \times \Big\{ \sum^N_{i=1}  \sup_{g_{l, t} \in \cA_{k, T}}  \Big(\int g_{l,t}(\varphi_N(x^i_{1:t}), \varphi_N(x_{t+1:T})) \mu (d x_{t+1:T} | \varphi^{-1}_N (x^i_{1:t})) \label{ineq:diff_temp} \\
		& \hspace{4cm} - \int g_{l,t}(\varphi_N(x^i_{1:t}), \varphi_N(x_{t+1:T})) \mu (d x_{t+1:T} | x^i_{1:t}) \Big)^2 \Big\}^{1/2}. \nonumber
	\end{align}
	
	By the property of $g_{l, t}$ and $\varphi_N$, we have
	\begin{align*}
		& \Big|  \int g_{l,t}(\varphi_N(x^i_{1:t}), \varphi_N(x_{t+1:T})) \mu (d x_{t+1:T} | \varphi^{-1}_N (x^i_{1:t}))  \\
		& \quad - \int g_{l,t}(\varphi_N(x^i_{1:t}), \varphi_N(x_{t+1:T})) \mu (d x_{t+1:T} | x^i_{1:t})  \\
		& \quad  - \Big(\int g_{l,t}(\varphi_N(x^i_{1:t}), x_{t+1:T}) \mu (d x_{t+1:T} | \varphi^{-1}_N (x^i_{1:t}))  \\
		& \qquad \quad - \int g_{l,t}(\varphi_N(x^i_{1:t}), x_{t+1:T}) \mu (d x_{t+1:T} | x^i_{1:t})  \Big) \Big| \\
		& \leq \int \Big| g_{l,t}(\varphi_N(x^i_{1:t}), \varphi_N(x_{t+1:T}))  -  g_{l,t}(\varphi_N(x^i_{1:t}), x_{t+1:T}) \Big| \mu (d x_{t+1:T} | \varphi^{-1}_N (x^i_{1:t})) \\
		& \quad + \int \Big| g_{l,t}(\varphi_N(x^i_{1:t}), \varphi_N(x_{t+1:T}))  -  g_{l,t}(\varphi_N(x^i_{1:t}), x_{t+1:T}) \Big| \mu (d x_{t+1:T} |x^i_{1:t}) \\
		& \leq \frac{L_\cA C_{b,a,T,d}}{N^q}.                  
	\end{align*}
	It yields
	\begin{align*}
		& \sup_{g_{l, t} \in \cA_{k, T}}  \Big(\int g_{l,t}(\varphi_N(x^i_{1:t}), \varphi_N(x_{t+1:T})) \mu (d x_{t+1:T} | \varphi^{-1}_N (x^i_{1:t})) \\
		& \qquad \qquad - \int g_{l,t}(\varphi_N(x^i_{1:t}), \varphi_N(x_{t+1:T})) \mu (d x_{t+1:T} | x^i_{1:t}) \Big)^2 \\
		& \leq 2 \Big[ \sup_{g_{l, t} \in \cA_{k, T}} \Big( \int g_{l,t}(\varphi_N(x^i_{1:t}), x_{t+1:T}) \mu (d x_{t+1:T} | \varphi^{-1}_N (x^i_{1:t})) \\
		& \hspace{2.5cm} - \int g_{l,t}(\varphi_N(x^i_{1:t}), x_{t+1:T}) \mu (d x_{t+1:T} | x^i_{1:t}) \Big) \Big]^2 + 2 \Big(\frac{L_\cA C_{b,a,T,d}}{N^q} \Big)^2 \\
		& \leq 2 \Big[ L_\cA \cW_1 \Big(\mu (\cdot | \varphi^{-1}_N (x^i_{1:t})), \mu (\cdot | x^i_{1:t}) \Big) \Big]^2 + 2 \Big(\frac{L_\cA C_{b,a,T,d}}{N^q} \Big)^2 \\
		& \leq 2 \Big(\frac{L_\cA C_{b,a,T,d} L_\mu }{N^q} \Big)^2 + 2 \Big(\frac{L_\cA C_{b,a,T,d}}{N^q} \Big)^2.
	\end{align*}
	The second inequality is from the strong duality of the {\it classic} Wasserstein-1 distance and the fact that $g_{l, t}$ is $L_\cA$-Lipschitz in $x_{t+1:T}$. The last inequality is due to Assumption \ref{assum:Lip_kernel} and \citet[Equation (3.5)]{backhoff2020estimating}. 
	
	Therefore, \eqref{ineq:diff_temp} becomes
	\begin{align*}
		& \Big| \E_{\sigma} \Big[ \sup_{g_{l, t} \in \cA_{k, T}} \sum^N_{i=1} \sigma_i \int g_{l,t}(\varphi_N(x^i_{1:t}), \varphi_N(x_{t+1:T})) \mu (d x_{t+1:T} | \varphi^{-1}_N (x^i_{1:t})) \Big] \\
		& \quad - \E_{\sigma} \Big[ \sup_{g_{l, t} \in \cA_{k, T}} \sum^N_{i=1} \sigma_i \int g_{l,t}(\varphi_N(x^i_{1:t}), \varphi_N(x_{t+1:T})) \mu (d x_{t+1:T} | x^i_{1:t}) \Big] \Big| \\
		& \leq N \frac{C_{b,a,T,d} (L_\cA  L_\mu + L_\cA)}{N^q},
	\end{align*} 
	which leads to \eqref{ineq:diff}. Here, $C_{b,a,T,d}$ is enlarged by a generic constant.

	Furthermore, thanks to Assumption \ref{assum:Lip_kernel}, we have
	\begin{align}
		& \frac{1}{N} \E_{\sigma} \Big[ \sup_{g_{l, t} \in \cA_{k, T}} \sum^N_{i=1} \sigma_i \int g_{l,t}(\varphi_N(x^i_{1:t}), \varphi_N(x_{t+1:T})) \mu (d x_{t+1:T} | x^i_{1:t}) \Big] \nonumber \\
		& = \frac{1}{N} \E_{\sigma} \Big[ \sup_{g_{l, t} \in \cA_{k, T}} \sum^N_{i=1} \sigma_i \int g_{l,t}(\varphi_N(x^i_{1:t}), \varphi_N(F_{t+1} (x^i_{1:t}, \varepsilon) )) \psi(d\varepsilon) \Big] \nonumber \\
		& \leq \frac{1}{N} \int \E_{\sigma} \Big[ \sup_{g_{l, t} \in \cA_{k, T}} \sum^N_{i=1} \sigma_i  g_{l,t}(\varphi_N(x^i_{1:t}), \varphi_N(F_{t+1} (x^i_{1:t}, \varepsilon) ))  \Big] \psi(d\varepsilon) \nonumber \\
		& \leq \int \frac{C_{w, k} L_\phi (C_{\vartheta, k} + C_{u, k} C_{b, a, T, d}) }{\sqrt{N}} \psi(d\varepsilon) \nonumber \\
		& \leq \frac{C_{w, k} L_\phi (C_{\vartheta, k} + C_{u, k} C_{b, a, T, d}) }{\sqrt{N}}. \label{ineq:Rad_intg}
	\end{align} 
	Here, the second inequality follows similarly as in Lemma \ref{lem:Rad_Gk}.
	
	{\bf Step 4}: Putting \eqref{ineq:Step1}, \eqref{ineq:h_g}, \eqref{ineq:h_integ}, \eqref{ineq:Rad1},  \eqref{ineq:Rad2}, \eqref{ineq:Rad3}, \eqref{ineq:diff}, and \eqref{ineq:Rad_intg} together, we obtain
	\begin{align*}
		& \cR_S( F( \Theta_k \circ \varphi_N); N) \\
		& \leq \frac{C_c \lambda_k}{\sqrt{N}} +  \frac{L_k (T-1)(2C_g + C_h) C_{w, k} L_\phi (C_{\vartheta, k} + C_{u, k} C_{b, a, T, d}) }{\sqrt{N}} \\
		& \quad + L_k(T-1)C_h \Big( \frac{(L_\cA L_\mu + L_\cA) C_{b, a, T, d}}{N^q} + \frac{C_{w, k} L_\phi (C_{\vartheta, k} + C_{u, k} C_{b, a, T, d})}{\sqrt{N}} \Big).
	\end{align*}
	Since the right-hand side does not depend on a specific sample $S$, this inequality also holds for $\cR( F( \Theta_k \circ \varphi_N); N)$.
\end{proof}

\begin{proof}[Proof of Theorem \ref{thm:Reg}]
	By Lemma \ref{lem:Rad_Gk}, $r(x; \delta)$ equals to $1/\sqrt{x}$ multiplied by a positive constant. Hence, Condition 2 in Lemma \ref{lem:Rad} holds and Lemma \ref{lem:Rad} follows. 
	
	By Lemma \ref{lem:Rad_Gk} and Lemma \ref{lem:Rad_F}, $\cR( F( \Theta_k \circ \varphi_N); N)$ and $  \cR( \cG_k \circ \varphi_N; N^{1 - qd(T-1)})$ in Lemma \ref{lem:Rad} converge to zero. Then Lemma \ref{lem:Ninf} holds.
	
	Thanks to the universal approximation theorem \cite[Theorem 1]{cybenko1989}, there exists a sequence of functions $g_k \in \Theta_k$ to approximate any continuous and bounded function $g$ on a given compact domain. It leads to the claim in Theorem \ref{thm:conv}.
\end{proof}

\end{document}